\documentclass[11pt,a4paper]{article}
\usepackage[T1]{fontenc}
\usepackage[utf8]{inputenc}
\usepackage[english]{babel}
\usepackage{amsmath}
\usepackage{amsfonts}
\usepackage{amssymb}
\usepackage{amsthm}
\usepackage{graphicx}
\usepackage{subcaption}
\usepackage{booktabs}
\usepackage[table]{xcolor}
\usepackage{multirow}
\usepackage{algorithm}
\usepackage{algorithmic}
\usepackage{cite}
\usepackage{url}
\usepackage{hyperref}
\usepackage{tikz}
\usepackage{pgfplots}
\pgfplotsset{compat=1.17}
\usepackage{geometry}
\usepackage{setspace}
\usepackage{titlesec}

\titlespacing*{\section}{0pt}{1ex}{0.5ex}
\titlespacing*{\subsection}{0pt}{0.8ex}{0.4ex}
\titlespacing*{\subsubsection}{0pt}{0.6ex}{0.3ex}
\titlespacing*{\paragraph}{0pt}{0.5ex}{1em}

\usepackage{enumitem}
\setlist{nosep, itemsep=0pt, parsep=0pt, topsep=2pt, partopsep=0pt}

\newtheorem{theorem}{Theorem}[section]
\newtheorem{definition}[theorem]{Definition}
\newtheorem{lemma}[theorem]{Lemma}

\hypersetup{
    colorlinks=true,
    linkcolor=blue,
    filecolor=magenta,
    urlcolor=cyan,
    citecolor=blue,
    pdftitle={Applying PMM2 to ARIMA Models with Non-Gaussian Innovations},
    pdfauthor={Serhii Zabolotnii},
    pdfkeywords={ARIMA, PMM2, non-Gaussian innovations, asymmetric distributions},
    bookmarksnumbered=true,
}

\newcommand{\thetavec}{\boldsymbol{\theta}}

\DeclareMathOperator{\Var}{Var}

\DeclareMathOperator{\RE}{RE}

\DeclareMathOperator{\tr}{tr}

\newcommand{\E}{\mathbb{E}}

\newcommand{\diffop}{\Delta^d}

\newcommand{\Normal}{\mathcal{N}}
\newcommand{\Gammadist}{\text{Gamma}}
\newcommand{\Lognormal}{\text{Lognormal}}
\newcommand{\Chisq}{\chi^2}

\title{\MakeUppercase{Applying the Polynomial Maximization Method to Estimate ARIMA Models with Asymmetric Non-Gaussian Innovations}}

\author{
\textbf{Serhii Zabolotnii}$^{a,b,c}$ \\[0.5em]
{\small $^{a}$\textit{Cherkasy State Business College, Cherkasy, Ukraine}} \\
{\small $^{b}$\textit{State Scientific Research Institute of Armament and Military Equipment}} \\
{\small \textit{Testing and Certification, Cherkasy, Ukraine}} \\
{\small $^{c}$\textit{Uzhhorod National University, Uzhhorod, Ukraine}}
}

\date{
\vspace{-1em}
{\small E-mail: \texttt{zabolotnii.serhii@csbc.edu.ua} \quad ORCID: 0000-0003-0242-2234}
}

\begin{document}

\maketitle

\begin{abstract}
\selectlanguage{english}

Classical estimators for ARIMA parameters (MLE, CSS, OLS) assume Gaussian innovations, an assumption frequently violated in financial and economic data exhibiting asymmetric distributions with heavy tails. We develop and validate the second-order polynomial maximization method (PMM2) for estimating ARIMA$(p,d,q)$ models with non-Gaussian innovations. PMM2 is a semiparametric technique that exploits higher-order moments and cumulants without requiring full distributional specification.

Monte Carlo experiments (128,000 simulations) across sample sizes $N = 100$, $200$, $500$ and $1000$ and four innovation distributions demonstrate that PMM2 substantially outperforms classical methods for asymmetric innovations. For ARIMA(1,1,0) with $N=500$, relative efficiency reaches 1.58--1.90 for Gamma, lognormal, and $\chi^2(3)$ innovations (37--47\% variance reduction). Under Gaussian innovations PMM2 matches OLS efficiency, avoiding the precision loss typical of robust estimators.

The method delivers major gains for moderate asymmetry ($|\gamma_3| \geq 0.5$) and $N \geq 200$, with computational costs comparable to MLE. Bridging classical statistical theory and modern computational data science, PMM2 provides an effective, computationally efficient alternative for time series with asymmetric innovations typical of financial markets, macroeconomic indicators, and industrial measurements. To facilitate seamless integration into modern data science workflows, the proposed method has been implemented and published as the open-source \texttt{EstemPMM} package on CRAN. Future extensions include seasonal SARIMA models, GARCH integration, and automatic order selection.
\end{abstract}

\begin{center}\rule{0.7\textwidth}{0.4pt}\end{center}
\noindent\textbf{Note added in version 3 (August 2026).} Two corrections and one
caveat. First, version~2 attributed to Pötscher (1991)~\cite{potscher1991noninvertibility} the
claim that Gaussian pseudo-likelihood can be inconsistent under distributional
misspecification. That reverses his result: his Theorem 2.1(b) establishes
consistency, and the cost of misspecification is efficiency rather than
consistency. The motivating discussion has been rewritten accordingly; the
contribution of the paper is unaffected, since it was always an efficiency
claim. Second, the relative-efficiency expression was printed as
$(4+2\gamma_4)/(4+2\gamma_4-\gamma_3^2)$, which is wrong by a factor of two in
both terms; the correct form is $(2+\gamma_4)/(2+\gamma_4-\gamma_3^2)$. All
numerical results in the paper were computed with the correct expression, so no
reported number changes. Third, the simulation results for specifications with
a moving-average component were obtained with an implementation that builds the
design matrix once from first-stage residuals. That is a one-step linearisation
of the innovation recursion and it costs a factor $1-\theta_1^2$ in efficiency;
re-solving the recursion at each candidate parameter value removes it. The
moving-average figures reported here are therefore conservative, and a revised
version with the corrected implementation is in preparation.
\begin{center}\rule{0.7\textwidth}{0.4pt}\end{center}


\vspace{1.5em}

\selectlanguage{english}

\section{Introduction}
\label{sec:introduction}

\subsection{Motivation}
\label{subsec:motivation}

Autoregressive integrated moving-average (ARIMA) models remain among the most widely used tools for analysing and forecasting time series across contemporary science. Since the seminal contribution of Box and Jenkins (1970), ARIMA models have been adopted in financial econometrics, macroeconomic forecasting, meteorological analysis, medical statistics, and numerous other domains~\cite{box2015time,hyndman2021forecasting}.

Classical estimation procedures for ARIMA models---maximum likelihood (MLE), conditional sum of squares (CSS), and ordinary least squares (OLS)---are built upon the central assumption of \textbf{Gaussian innovations}. While this assumption guarantees several desirable statistical properties such as asymptotic efficiency and tractable inference, the explosion of data in modern computational statistics and data science consistently reveals violations of normality in real-world applications.

Recent studies provide compelling evidence of non-Gaussian behaviour across different classes of time series:

\begin{itemize}
    \item \textbf{Financial time series.} Equity returns, exchange rates, and volatility exhibit asymmetric distributions with heavy tails. Even after accounting for volatility dynamics via GARCH models, heavy tails persist~\cite{kantelhardt2002multifractal,kim2012approximation}. A recent study of the Korean stock market confirms persistent heavy tails even after controlling for crisis episodes and volatility clustering~\cite{kim2019fat}.

    \item \textbf{Economic indicators.} Commodity prices, inflation data, and trade volumes are characterised by pronounced asymmetry. An investigation covering 15 economies over 1851--1913 found a strong link between commodity-price asymmetry and inflation, with up to 48\% of inflation variability explained by commodity-price movements~\cite{gerlach2024commodity}.

    \item \textbf{Environmental and meteorological data.} Pollution readings, precipitation, temperature anomalies, and solar activity are frequently skewed and display extreme outcomes. Verma et al. (2025) document heavy tails in solar flare data and discuss theoretical limits of forecasting under heavy-tailed distributions~\cite{verma2025optimal}.

    \item \textbf{High-frequency financial data.} Mixed-stable models applied to DAX constituents at 10-second intervals uncover 43--82\% zero returns (stagnation effects), necessitating specialised modelling techniques~\cite{slezak2023application,dedomenico2023modeling}.
\end{itemize}

\subsection{Limitations of Classical Methods}
\label{subsec:limitations}

When the Gaussian assumption fails, classical estimators for ARIMA parameters face several shortcomings:

\paragraph{Loss of efficiency under misspecification.} Pötscher (1991)~\cite{potscher1991noninvertibility} shows that the Gaussian pseudo-likelihood estimator of an ARMA model remains consistent when the innovation distribution is misspecified (his Theorem 2.1(b)); the difficulties he identifies concern noninvertible boundary cases rather than consistency itself. What misspecification costs is efficiency, and that cost grows with the asymmetry of the innovations. This is the gap the present method addresses. In a related but distinct setting, Fan, Qi, and Xiu (2014)~\cite{fan2014quasi} study non-Gaussian quasi-maximum likelihood for GARCH models and propose a two-step procedure whose second stage recovers efficiency; their results concern conditional-variance models and do not transfer automatically to ARIMA, but the two-stage architecture is the closest analogue to the estimator developed here.

\paragraph{Loss of statistical efficiency.} Even when estimators remain consistent, their variance may be markedly inflated relative to estimators tailored to the true distribution. Zhang and Sin (2012) establish that the limiting laws are mixtures of stable and Gaussian processes for near-unit-root AR processes with $\alpha$-stable noise, underlining the difficulties introduced by heavy tails combined with near-integration~\cite{zhang2012maximum}.

\paragraph{Degraded forecast accuracy.} Li et al. (2020) document sizeable forecast errors for traditional ARIMA models at high frequencies because financial data display irregular fluctuations that require alternative approaches~\cite{li2020forecasting}. Dowe et al. (2025) show that hybrid ARFIMA-ANN strategies better capture complex non-Gaussian dynamics in financial and environmental data while leveraging the minimum message length principle for model selection~\cite{dowe2025novel}.

\paragraph{Misleading confidence intervals.} Ledolter (1989) demonstrates that ignoring outliers inflates mean-squared forecast error and biases parameter estimates in stock-price applications~\cite{ledolter1989inference}. Consequently, uncertainty bands become understated or overstated, undermining decision making.

\subsection{Existing Approaches: Brief Overview}
\label{subsec:existing_approaches}

In response to non-Gaussianity in time series, several alternative approaches have been developed by the research community:

\paragraph{Robust estimation methods (M-estimators).} Initiated by the classic work of Huber (1964)~\cite{huber1964robust}, M-estimators minimise robust loss functions that are less sensitive to outliers and heavy tails. Muler et al. (2009) introduced BIP-ARMA models with MM-estimates that avoid outlier propagation through bounded residuals, achieving consistency and asymptotic normality with efficiency comparable to MLE under normality~\cite{muler2009robust}. Reisen et al. (2024) proposed an M-Whittle estimator with established consistency property that performs well with outliers and heavy-tailed noise~\cite{reisen2024robust}.

\paragraph{Quantile regression and LAD methods.} Katsouris (2023) provided a comprehensive review of quantile regression models for time series, covering stationary and non-stationary cases, with Bahadur representations for quantile processes and uniform inference in quantile threshold regression~\cite{katsouris2023quantile}. For ARMA models with infinite variance, Peng \& Yao (2003), Ling (2005), and Zhu \& Ling (2015) proposed weighted least absolute deviation estimation (WLADE) that is asymptotically normal and unbiased with standard root-n convergence rate even in the absence of finite variance~\cite{peng2003least,ling2005self,zhu2015model}.

\paragraph{Heavy-tailed specifications.} Modifying classical ARIMA models by replacing Gaussian innovations with heavy-tailed distributions (Student-t, Generalized Error Distribution, $\alpha$-stable distributions) allows better modelling of extreme events. Wong et al. (2009) developed a Student-t mixture autoregressive model with higher flexibility compared to Gaussian MAR, where degrees of freedom are random variables, using an EM algorithm for parameter estimation in a Bayesian framework~\cite{wong2009student}. A recent 2024 study found that skewed GED is most effective for financial time series compared to normal, Student-t, GED, and Skewed Student-t distributions based on goodness-of-fit metrics~\cite{ampadu2024comparative}. Additionally, recent advances in multivariate contexts, such as the matrix-variate shifted generalized asymmetric Laplace distribution introduced by Pereira et al. (2026)~\cite{pereira2026matrix}, highlight the ongoing focus on flexible modelling of asymmetric and heavy-tailed data structures.

\paragraph{Bayesian approaches.} Graves et al. (2015) proposed a systematic approach to Bayesian inference for ARFIMA models with a novel approximate likelihood for efficient inference of parameters in long-memory processes, allowing innovations from a broad class including $\alpha$-stable and t-distributions~\cite{graves2015efficient}. Bayesian methods also integrate uncertainty across all parameters, providing full posterior inference instead of point estimates.

\paragraph{Transformer-based architectures.} Recent research shows that Transformer architectures adapted for time series can implicitly model non-Gaussian characteristics through the attention mechanism~\cite{zhou2024comprehensive}. Hybrid models combining traditional ARIMA with attention-based neural networks demonstrate improved forecasting accuracy for financial time series~\cite{liu2024arima}. Foundation models such as TimeGPT-1 and Chronos represent a new direction for generalised forecasting across diverse datasets without specific training~\cite{godfried2024advances}. However, these models are:
\begin{itemize}
    \item \textbf{Data-hungry}: require thousands of observations for stable training, whereas classical ARIMA works with $N \geq 100$
    \item \textbf{Black-box}: lack of interpretable parameters complicates econometric inference
    \item \textbf{Computationally intensive}: training requires GPUs and takes hours, compared to classical models.
\end{itemize}

It is important to note that deep models are optimised for minimising prediction error, not for efficient \textit{parameter estimation} of the underlying data-generating process.

\paragraph{Comparative analysis of approaches.} Each of the described methods has specific advantages and limitations. Robust methods (M-estimators) provide resilience to outliers and retain asymptotic efficiency under normality, but may lose accuracy under moderate deviations from Gaussianity without extreme outliers. Quantile regression and LAD methods provide a complete picture of the distribution and work without assuming finite variance, but are not optimised for central parameter estimates and mean forecasts. Heavy-tailed specifications explicitly model extreme events and provide interpretable distribution-shape parameters, but require a priori choice of distributional family, which may lead to model misspecification. Bayesian methods integrate parameter uncertainty and allow flexible innovation specifications, but are computationally intensive and require choice of prior distributions. Finally, Transformer-based models achieve high forecasting accuracy and adapt to complex patterns, but do not provide interpretable parametrisation of the data-generating process.

Thus, there is a need for an approach that combines the statistical efficiency of classical methods, interpretability of parameters, and robustness to deviations from normality without requiring specification of a particular distributional family.

\subsection{Polynomial Maximization Method: An Alternative Paradigm}
\label{subsec:pmm_intro}

The polynomial maximization method (PMM), developed by Ukrainian scholar Yu.P. Kunchenko, offers an alternative philosophy for statistical estimation~\cite{kunchenko2002polynomial}. Unlike classical maximum likelihood, which requires full specification of the probability density, PMM is built on \textbf{partial probabilistic parametrisation} via higher-order moments and cumulants.

The method centres on maximising a stochastic polynomial of order $S$ with respect to the model parameters. Rather than maximising the full likelihood, PMM maximises a sample statistic in a neighbourhood of the true parameter values~\cite{kunchenko2002polynomial}.

PMM has been successfully applied to diverse estimation problems:

\begin{itemize}
    \item \textbf{Linear regression.} Zabolotnii et al. (2018) demonstrate PMM2 for linear regression with asymmetric errors, achieving a 15--35\% variance reduction relative to OLS for Gamma and lognormal distributions~\cite{zabolotnii2018polynomial}.

    \item \textbf{Polynomial regression.} Zabolotnii et al. (2021) extend the method to polynomial regression with an exponential-power (generalised Gaussian) distribution, confirming efficiency gains through Monte Carlo and bootstrap experiments~\cite{zabolotnii2021estimating}.

    \item \textbf{Signal processing.} Palahin and Juhár (2016) apply PMM to joint estimation of signal parameters under non-Gaussian noise, showing that nonlinear processing through third- and higher-order cumulants reduces joint-estimation variance compared with conventional techniques~\cite{palahin2016joint}.

    \item \textbf{Metrological measurements.} Warsza and Zabolotnii (2017, 2018) use PMM to estimate measurement parameters with non-Gaussian symmetric and asymmetric data, designing the PMM3 procedure for symmetric distributions~\cite{warsza2017polynomial,zabolotnii2020estimation}.
\end{itemize}

It is worth noting that the PMM method is positioned between the classical method of moments and maximum likelihood. Unlike Hansen's (1982) generalised method of moments (GMM)~\cite{hansen1982gmm}, which minimises a weighted sum of squared deviations between sample and population moments, PMM maximises a stochastic polynomial constructed from higher-order moments or cumulants.

\subsection{Research Gap and Contribution}
\label{subsec:research_gap}

Despite successful applications in regression and signal-processing settings, PMM has not been systematically explored for estimating ARIMA models with non-Gaussian innovations. Several research gaps remain:

\paragraph{Limited development of moment--cumulant methods for time series.} Although higher-order moments and cumulants are widely used in signal processing and spectral analysis, their application to parameter estimation in time-series models is limited. Most non-Gaussian ARIMA approaches rely on robust loss functions or explicit distributional assumptions rather than exploiting moment--cumulant representations.

\paragraph{Insufficient focus on asymmetric innovations.} The literature on non-Gaussian ARIMA models largely concentrates on symmetric heavy-tailed distributions (Student-$t$, GED). Asymmetric distributions---the core target of PMM2---receive far less attention despite their prevalence in financial returns and economic indicators.

\paragraph{Methodological divide across research communities.} Kunchenko's method, although supported by solid theory and successful applications in Eastern Europe, remains little known in Western time-series econometrics. Bridging these communities offers an opportunity to cross-fertilise statistical methodology.

\paragraph{Absence of comparative efficiency studies.} Comparative work typically evaluates MLE, M-estimators, least absolute deviation, or quantile regression. For ARIMA models there is no benchmark assessing the efficiency of moment--cumulant methods such as PMM relative to these alternatives.

This study fills these gaps by adapting PMM2 to ARIMA models, deriving asymptotic properties of the resulting estimator, and delivering comprehensive simulation and empirical evidence. We also provide implementation guidelines, including moment-calibration rules and decision heuristics for deciding when PMM2 should replace classical estimators. Furthermore, to ensure broad accessibility and computational reproducibility, the method has been implemented in the \texttt{EstemPMM} R package, officially published on the Comprehensive R Archive Network (CRAN).

\section{Methodology}
\label{sec:methodology}

This section presents the full methodology for applying the second-order polynomial maximization method (PMM2) to estimate ARIMA models with non-Gaussian innovations. We first outline the ARIMA model and classical estimators, then review the theoretical foundations of PMM, adapt the method to the time-series setting, and derive an implementable algorithm together with asymptotic results.

\subsection{ARIMA Models: Foundations and Classical Estimation}
\label{subsec:arima_basics}

Consider the standard ARIMA$(p,d,q)$ specification for a time series $\{y_t\}_{t=1}^T$, where the $d$-th difference $z_t = \diffop y_t = (1-B)^d y_t$ follows a stationary and invertible ARMA$(p,q)$ model:
\begin{equation}
\label{eq:arma}
\Phi(B) z_t = \Theta(B) \varepsilon_t,
\end{equation}
with $B$ the backshift operator, $\Phi(B) = 1 - \phi_1 B - \cdots - \phi_p B^p$ and $\Theta(B) = 1 + \theta_1 B + \cdots + \theta_q B^q$ the autoregressive and moving-average polynomials, and $\{\varepsilon_t\}$ an i.i.d. innovation sequence satisfying $\E[\varepsilon_t]=0$ and $\Var(\varepsilon_t)=\sigma^2$. Characteristic roots of $\Phi(z)=0$ and $\Theta(z)=0$ lie outside the unit circle, ensuring stationarity and invertibility.

\paragraph{Classical estimators.} Let $\boldsymbol{\theta} = (\phi_1,\ldots,\phi_p,\theta_1,\ldots,\theta_q)^\top$ denote the $k=p+q$ parameters. Two benchmark estimators are used for comparison:

\textit{Ordinary least squares (OLS)} applied to the AR component:
\begin{equation}
\label{eq:ols}
\hat{\boldsymbol{\phi}}_{\text{OLS}} = (\mathbf{X}^\top \mathbf{X})^{-1} \mathbf{X}^\top \mathbf{z},
\end{equation}
where $\mathbf{z} = (z_{p+1},\ldots,z_T)^\top$ and $\mathbf{X}$ collects lagged values.

\textit{Maximum likelihood (MLE)} under $\varepsilon_t \sim \Normal(0,\sigma^2)$:
\begin{equation}
\label{eq:mle}
\hat{\boldsymbol{\theta}}_{\text{MLE}} = \arg\max_{\boldsymbol{\theta}} \mathcal{L}(\boldsymbol{\theta}\mid \mathbf{y}) = \arg\max_{\boldsymbol{\theta}} \left\{ -\frac{T}{2}\log(2\pi\sigma^2) - \frac{1}{2\sigma^2} \sum_{t=1}^T \varepsilon_t^2(\boldsymbol{\theta}) \right\}.
\end{equation}

Under Gaussian innovations, MLE attains asymptotic efficiency: $\sqrt{T}(\hat{\boldsymbol{\theta}}_{\text{MLE}} - \boldsymbol{\theta}_0) \xrightarrow{d} \Normal(0, \mathbf{I}^{-1}(\boldsymbol{\theta}_0))$, where $\mathbf{I}(\boldsymbol{\theta}_0)$ is the Fisher information matrix. With non-Gaussian innovations, however, MLE loses optimality and OLS remains consistent but inefficient, motivating alternative estimators adapted to non-Gaussian distributions.

\subsection{Theoretical Foundations of the Polynomial Maximization Method}
\label{subsec:pmm_theory}

\subsubsection{Stochastic polynomials}

The polynomial maximization method relies on stochastic polynomials, i.e., polynomial functions of random variables whose coefficients depend on the model parameters. The approach was devised for situations where the probabilistic properties of the data depart markedly from the Gaussian law.

\begin{definition}[Stochastic polynomial of order $S$]
For random variables $y_v$, $v = 1,\ldots,N$, and a parameter vector $\mathbf{a}$, the order-$S$ stochastic polynomial is defined as
\begin{equation}
\label{eq:stochastic_polynomial_general}
L_{SN} = \sum_{v=1}^{N} \sum_{i=1}^{S} \phi_i(y_v) \int k_{iv}(a)\, dz - \sum_{i=1}^{S} \sum_{v=1}^{N} \int \Psi_{iv}\, k_{iv}(a)\, dz,
\end{equation}
where $\phi_i(y_v)$ are basis functions, $k_{iv}(a)$ are weights depending on $a$, and $\Psi_{iv} = \E[\phi_i(y_v)]$ are twice-differentiable expectations with respect to $a$.
\end{definition}

\paragraph{Fundamental properties.}

The stochastic polynomial $L_{SN}$ in~\eqref{eq:stochastic_polynomial_general} satisfies two key properties~\cite{kunchenko2002polynomial}:

\begin{enumerate}
    \item For any order $S$, as the sample size $N \to \infty$ the polynomial $L_{SN}$, viewed as a function of $a$, attains its maximum near the true value of $a$.
    \item Across different samples the deviation of the maximiser of $L_{SN}$ from the true $a$ has minimal variance for the chosen order $S$.
\end{enumerate}

Analogously to maximum likelihood, the estimator of $a$ solves
\begin{equation}
\label{eq:pmm_estimation_eq}
\frac{d}{da} L_{SN} \bigg|_{a=\hat{a}} = \sum_{i=1}^{S} \sum_{v=1}^{N} k_{iv} \big[\phi_i(y_v) - \Psi_{iv}\big] \bigg|_{a=\hat{a}} = 0.
\end{equation}

\paragraph{Optimal coefficients and linear system.}

The optimal weights $k_{iv}$ that maximise~\eqref{eq:stochastic_polynomial_general} are obtained by solving
\begin{equation}
\label{eq:optimal_coefficients}
\sum_{j=1}^{S} k_{jv} F_{(i,j)v} = \frac{d}{da} \Psi_{iv}, \quad i=1,\ldots,S,\quad v=1,\ldots,N,
\end{equation}
where $F_{(i,j)v} = \Psi_{(i,j)v} - \Psi_{iv} \Psi_{jv}$ and $\Psi_{(i,j)v} = \E[\phi_i(y_v)\phi_j(y_v)]$.

\paragraph{Vector parameters and multi-parameter estimation.}

To estimate a parameter vector $\thetavec = (a_0,\ldots,a_{Q-1})^\top$, one uses $Q$ polynomials $L_{SN}^{(p)}$, $p=0,\ldots,Q-1$, each of the form~\eqref{eq:stochastic_polynomial_general} for the corresponding component $a_p$.

Each polynomial $L_{SN}^{(p)}$, treated as a function of $a_p$ with the remaining parameters fixed, attains its maximum near the true value of $a_p$ as $N \to \infty$. The resulting estimators solve
\begin{equation}
\label{eq:vector_estimation}
f_{SN}^{(p)}(y_v, x_v) = \sum_{i=1}^{S} \sum_{v=1}^{N} k_{iv}^{(p)} \big[\phi_i(y_v) - \Psi_{iv}\big] \bigg|_{a_p=\hat{a}_p} = 0, \quad p=0,\ldots,Q-1.
\end{equation}

\subsubsection{PMM for asymmetric distributions}

Consider a linear multiple regression model with asymmetric disturbances~\cite{zabolotnii2018polynomial}. Power transformations serve as basis functions whose expectations are moments of the corresponding order. For observations $\{y_v\}_{v=1}^N$
\begin{equation}
\label{eq:multiple_regression_model}
y_v = \boldsymbol{x}_v^\top \boldsymbol{\theta} + \xi_v,\qquad \boldsymbol{x}_v = (1, x_{1,v}, \ldots, x_{Q-1,v})^\top,
\end{equation}
where $\boldsymbol{\theta} = (a_0,\ldots,a_{Q-1})^\top$ and the disturbance satisfies
\[
\E[\xi_v]=0,\quad \E[\xi_v^2]=\mu_2>0,\quad \E[\xi_v^3]=\mu_3\neq 0,\quad \E[\xi_v^4]=\mu_4<\infty.
\]
Let $\eta_v(\boldsymbol{\theta}) = \boldsymbol{x}_v^\top \boldsymbol{\theta}$ and $\mathbf{X} = [\boldsymbol{x}_1,\ldots,\boldsymbol{x}_N]^\top$.

\paragraph{PMM1: linear polynomial and equivalence to OLS.}

For $S=1$ choose $\phi_1(y_v) = y_v$, yielding $\Psi_{1v} = \E[y_v] = \eta_v(\boldsymbol{\theta})$. The covariance $F_{(1,1)v} = \mu_2$ is constant, and from~\eqref{eq:optimal_coefficients} we obtain $k_{1,v}^{(p)} = x_{p,v}/\mu_2$ with $x_{0,v} \equiv 1$. The first-order conditions become
\begin{equation}
\label{eq:pmm1_system}
\sum_{v=1}^{N} x_{p,v} \left[ y_v - \eta_v(\boldsymbol{\theta}) \right] = 0,\qquad p = 0,\ldots,Q-1,
\end{equation}
which coincide with the normal equations $\mathbf{X}^\top \mathbf{X}\,\boldsymbol{\theta} = \mathbf{X}^\top \boldsymbol{y}$. Hence PMM1 reproduces the OLS estimator
\begin{equation}
\label{eq:pmm1_solution}
\hat{\boldsymbol{\theta}}_{\mathrm{PMM1}} = (\mathbf{X}^\top \mathbf{X})^{-1} \mathbf{X}^\top \boldsymbol{y},
\end{equation}
optimal only when errors are Gaussian.

\paragraph{PMM2: second-order stochastic polynomial.}

To accommodate asymmetry we form a stochastic polynomial with basis functions
\begin{align}
\phi_1(y_v) &= y_v, & \Psi_{1v} &= \eta_v(\boldsymbol{\theta}), \\
\phi_2(y_v) &= y_v^2, & \Psi_{2v} &= \eta_v^2(\boldsymbol{\theta}) + \mu_2.
\end{align}
The matrices $F_{(i,j)v} = \Psi_{(i,j)v} - \Psi_{iv}\Psi_{jv}$ depend on central moments up to order four, leading to the optimal coefficients
\begin{equation}
\label{eq:pmm2_k_general}
k_{1,v}^{(p)} = \frac{\mu_4 - \mu_2^2 + 2\mu_3 \eta_v(\boldsymbol{\theta})}{\Delta}\, x_{p,v}, \qquad
k_{2,v}^{(p)} = -\frac{\mu_3}{\Delta}\, x_{p,v},
\end{equation}
where
\begin{equation}
\label{eq:delta_definition}
\Delta = \mu_2\big(\mu_4 - \mu_2^2\big) - \mu_3^2 > 0.
\end{equation}
The estimating equations become
\begin{equation}
\label{eq:pmm2_system_general}
g_p(\boldsymbol{\theta}) = \sum_{v=1}^{N} x_{p,v} \left\{ \frac{\mu_4 - \mu_2^2 + 2\mu_3 \eta_v(\boldsymbol{\theta})}{\Delta} \big[y_v - \eta_v(\boldsymbol{\theta})\big] - \frac{\mu_3}{\Delta} \big[y_v^2 - \eta_v^2(\boldsymbol{\theta}) - \mu_2\big] \right\} = 0,
\end{equation}
for $p = 0,\ldots,Q-1$, which collapses to the OLS system when $\mu_3 = 0$.

Multiplying~\eqref{eq:pmm2_system_general} by $\Delta$ and grouping terms by powers of $\eta_v(\boldsymbol{\theta})$ yields the quadratic system
\begin{equation}
\label{eq:pmm2_quadratic_form}
\sum_{v=1}^{N} x_{p,v} \left[ A_2 \eta_v^2(\boldsymbol{\theta}) + B_{2,v}\,\eta_v(\boldsymbol{\theta}) + C_{2,v} \right] = 0,\quad p = 0,\ldots,Q-1,
\end{equation}
with coefficients
\begin{equation}
\label{eq:pmm2_coeffs}
A_2 = \mu_3,\qquad B_{2,v} = (\mu_4 - \mu_2^2) - 2\mu_3 y_v,\qquad C_{2,v} = \mu_3 y_v^2 - y_v (\mu_4 - \mu_2^2) - \mu_2 \mu_3.
\end{equation}

\paragraph{Matrix form and Newton--Raphson step.}

Define $\boldsymbol{g}(\boldsymbol{\theta}) = (g_0(\boldsymbol{\theta}),\ldots,g_{Q-1}(\boldsymbol{\theta}))^\top$ and
\begin{equation}
\label{eq:lambda_definition}
\lambda_v(\boldsymbol{\theta}) = \frac{2\mu_3 [y_v - \eta_v(\boldsymbol{\theta})] - (\mu_4 - \mu_2^2)}{\Delta}.
\end{equation}
The Jacobian becomes
\begin{equation}
\label{eq:pmm2_jacobian}
\mathbf{J}_{\mathrm{PMM2}}(\boldsymbol{\theta}) = \sum_{v=1}^{N} \lambda_v(\boldsymbol{\theta})\, \boldsymbol{x}_v \boldsymbol{x}_v^\top,
\end{equation}
and the Newton--Raphson step updates
\begin{equation}
\label{eq:pmm_newton}
\boldsymbol{\theta}^{(m+1)} = \boldsymbol{\theta}^{(m)} - \mathbf{J}_{\mathrm{PMM2}}^{-1}\big(\boldsymbol{\theta}^{(m)}\big)\, \boldsymbol{g}\big(\boldsymbol{\theta}^{(m)}\big),
\end{equation}
with the OLS estimator~\eqref{eq:pmm1_solution} serving as a convenient initial value $\boldsymbol{\theta}^{(0)}$.

\paragraph{Adaptive procedure.}

In practice the moments $\mu_2$, $\mu_3$, and $\mu_4$ are unknown and are replaced by sample estimates based on current residuals. A standard iterative routine proceeds as follows:
\begin{enumerate}
    \item \textbf{Step 1:} Compute the OLS estimator $\hat{\boldsymbol{\theta}}^{\mathrm{OLS}}$ and residuals $\hat{\xi}_v^{(0)} = y_v - \boldsymbol{x}_v^\top \hat{\boldsymbol{\theta}}^{\mathrm{OLS}}$.
    \item \textbf{Step 2:} For iteration $m$ update the moments
    \[
    \hat{\mu}_r^{(m)} = \frac{1}{N} \sum_{v=1}^{N} \big(\hat{\xi}_v^{(m)}\big)^r,\quad r \in \{2,3,4\},
    \]
    and compute the excess kurtosis $\hat{\gamma}_4^{(m)} = \hat{\mu}_4^{(m)} / (\hat{\mu}_2^{(m)})^2 - 3$ and skewness $\hat{\gamma}_3^{(m)} = \hat{\mu}_3^{(m)} / (\hat{\mu}_2^{(m)})^{3/2}$.
    \item \textbf{Step 3:} If $|\hat{\gamma}_3^{(m)}| < 0.1$, retain the OLS estimator. Otherwise solve~\eqref{eq:pmm2_system_general} via~\eqref{eq:pmm_newton} using $\hat{\mu}_r^{(m)}$ to obtain $\boldsymbol{\theta}^{(m+1)}$.
    \item \textbf{Step 4:} Update residuals $\hat{\xi}_v^{(m+1)} = y_v - \boldsymbol{x}_v^\top \boldsymbol{\theta}^{(m+1)}$ and iterate steps 2--4 until convergence.
\end{enumerate}

\subsubsection{Asymptotic variances and efficiency of PMM estimators}

\paragraph{Information matrix.}

Analytical expressions for the variances of PMM estimators rely on the matrix of acquired information for order-$S$ stochastic polynomials:
\begin{equation}
\label{eq:fisher_information_matrix}
J_{SN}^{(p,q)} = \sum_{v=1}^{N} \sum_{i=1}^{S} \sum_{j=1}^{S} k_{iv}^{(p)} k_{jv}^{(q)} F_{(i,j)v} = \sum_{v=1}^{N} \sum_{i=1}^{S} k_{i,v}^{(p)} \frac{\partial}{\partial a_q} \Psi_{iv}, \quad p,q=0,\ldots,Q-1.
\end{equation}

This quantity is conceptually analogous to Fisher information. In the asymptotic regime ($N \to \infty$) the variance matrix of PMM estimators equals the inverse of~\eqref{eq:fisher_information_matrix}:
\begin{equation}
\label{eq:variance_matrix}
\mathbf{V}_{\text{PMM}S}(\boldsymbol{\theta}) = \big[\mathbf{J}_S(\boldsymbol{\theta})\big]^{-1}.
\end{equation}

\paragraph{Convergence to the Rao--Cramér bound.}

A key feature of PMM is that as $S \to \infty$ the estimator approaches the minimum-variance unbiased estimator and attains the Rao--Cramér bound. For $S=2$---the focus for asymmetric innovations---the estimator already exploits skewness and kurtosis, delivering notable efficiency gains over OLS without requiring full density specification.

For the scalar case ($Q=1$) the ratio of asymptotic variances between PMM2 and OLS simplifies to
\begin{equation}
\label{eq:re_pmm2_ols}
\RE_{\text{PMM2/OLS}} = \frac{2 + \gamma_4}{2 + \gamma_4 - \gamma_3^2},
\end{equation}
where $\gamma_3$ and $\gamma_4$ denote the standardised skewness and excess kurtosis of the innovations. This expression highlights the quadratic impact of skewness on efficiency gains.

\subsection{PMM2 for ARIMA Models: Method Adaptation}
\label{subsec:pmm2_arima}

\subsubsection{Motivation and approximation principle}

Adapting PMM2 to ARIMA processes hinges on pre-stationarisation: differencing of order $d$ enables the constructions from Section~\ref{subsec:pmm_theory} to operate on the stationary series $z_t$. After this transformation the baseline method reliably recovers innovations and the PMM2 correction refines estimates under asymmetry. The approach preserves PMM2's sensitivity to higher moments while avoiding complex recursions for pseudo-regressors. Instead we implement a simple two-stage routine:

\begin{enumerate}
    \item \textbf{Baseline step.} Estimate ARIMA$(p,d,q)$ via a standard method (CSS or ML) and treat the resulting residuals as empirical innovations.
    \item \textbf{PMM2 correction.} Fix the design matrix formed at the baseline step and apply the second-order polynomial adjustment to incorporate innovation skewness and kurtosis.
\end{enumerate}

\subsubsection{Pseudo-regressor construction}

Let
\begin{equation}
\label{eq:differenced_series}
z_t = \diffop y_t,\qquad t=d+1,\ldots,T,\qquad n = T-d,
\end{equation}
denote the stationarised series. After estimating ARIMA$(p,d,q)$ in the first stage we obtain residuals $\widehat{\varepsilon}_t^{\text{CSS}}$. For the effective length $n_\text{eff} = n - m$ with $m = \max(p,q)$ define the regressors
\begin{equation}
\label{eq:design_row}
\mathbf{x}_t = \big(z_{t-1},\ldots,z_{t-p},\widehat{\varepsilon}_{t-1}^{\text{CSS}},\ldots,\widehat{\varepsilon}_{t-q}^{\text{CSS}}\big)^\top,\qquad t = m+1,\ldots,n.
\end{equation}
If an intercept is present, append a column of ones. The resulting matrix $\mathbf{X} = (\mathbf{x}_{m+1},\ldots,\mathbf{x}_{n})^\top$ is parameter-free, reducing subsequent optimisation to the problem in Section~\ref{subsec:pmm_theory}.

\subsubsection{Moment calibration and stochastic polynomial}

Using residuals from the baseline step compute the central moments
\begin{equation}
\label{eq:css_moments}
\hat{\mu}_k = \frac{1}{n_\text{eff}} \sum_{t=m+1}^{n} \big(\widehat{\varepsilon}_t^{\text{CSS}} - \bar{\varepsilon}\big)^k,\qquad k=2,3,4,
\end{equation}
where $\bar{\varepsilon}$ is the sample mean. Analogously to the basic PMM2 we define
\begin{equation}
\label{eq:delta_hat}
\hat{\Delta} = \hat{\mu}_2(\hat{\mu}_4 - \hat{\mu}_2^2) - \hat{\mu}_3^2.
\end{equation}
With $\boldsymbol{\theta} = (\phi_1,\ldots,\phi_p,\theta_1,\ldots,\theta_q)^\top$ and $\eta_t(\boldsymbol{\theta}) = \mathbf{x}_t^\top \boldsymbol{\theta}$, the stochastic polynomial becomes
\begin{equation}
\label{eq:simplified_score}
g_j(\boldsymbol{\theta}) = \sum_{t=m+1}^{n} x_{j,t}
\left[
\frac{\hat{\mu}_4 - \hat{\mu}_2^2 + 2\hat{\mu}_3 \eta_t(\boldsymbol{\theta})}{\hat{\Delta}}\big(z_t - \eta_t(\boldsymbol{\theta})\big)
-\frac{\hat{\mu}_3}{\hat{\Delta}}\Big(z_t^2 - \eta_t^2(\boldsymbol{\theta}) - \hat{\mu}_2\Big)
\right]=0,
\end{equation}
where $x_{j,t}$ denotes the $j$th component of $\mathbf{x}_t$. System~\eqref{eq:simplified_score} mirrors~\eqref{eq:pmm2_system_general} for a fixed design matrix. In the symmetric limiting case ($\hat{\mu}_3 = 0$) it reduces to weighted least squares with weight $\hat{\mu}_2^{-1}$.

\subsection{PMM2 Estimation Algorithm for ARIMA}
\label{subsec:algorithm}

\begin{algorithm}[H]
\caption{Simplified PMM2 estimator for ARIMA$(p,d,q)$}
\label{alg:pmm2_arima}
\begin{algorithmic}[1]
\REQUIRE Time series $\{y_t\}_{t=1}^T$, orders $(p,d,q)$, choice of initial estimator (CSS or ML)
\ENSURE Parameter vector $\hat{\boldsymbol{\theta}}_{\text{PMM2}}$, moment estimates $\hat{\mu}_2,\hat{\mu}_3,\hat{\mu}_4$

\STATE \textbf{Differencing.} Compute $z_t = \diffop y_t$ using~\eqref{eq:differenced_series}.

\STATE \textbf{Baseline estimation.} Obtain $\hat{\phi}_j^{\text{CSS}}$, $\hat{\theta}_k^{\text{CSS}}$, and residuals $\widehat{\varepsilon}_t^{\text{CSS}}$ via the chosen standard method.

\STATE \textbf{Design matrix.} Form $\mathbf{X}$ from rows~\eqref{eq:design_row} and the response vector $\mathbf{z} = (z_{m+1},\ldots,z_n)^\top$.

\STATE \textbf{Moment evaluation.} Compute $\hat{\mu}_2,\hat{\mu}_3,\hat{\mu}_4$ using~\eqref{eq:css_moments} and $\hat{\Delta}$ from~\eqref{eq:delta_hat}.

\STATE \textbf{Initialization.} Set $\boldsymbol{\theta}^{(0)} = (\hat{\phi}_1^{\text{CSS}},\ldots,\hat{\phi}_p^{\text{CSS}},\hat{\theta}_1^{\text{CSS}},\ldots,\hat{\theta}_q^{\text{CSS}})^\top$.

\STATE \textbf{Polynomial optimization.} Apply the iterative PMM2 solver for the fixed design $\mathbf{X}$ (see~\eqref{eq:simplified_score}), stopping when $\|\boldsymbol{\theta}^{(k)} - \boldsymbol{\theta}^{(k-1)}\|$ and the norm of the moment conditions fall below a predefined tolerance.

\STATE \textbf{Residual reconstruction.} Recompute innovations by passing the final parameters through the ARIMA model; retain them for diagnostic checks.

\RETURN $\hat{\boldsymbol{\theta}}_{\text{PMM2}} = \boldsymbol{\theta}^{(k_\star)}$ together with $\hat{\mu}_2,\hat{\mu}_3,\hat{\mu}_4$.
\end{algorithmic}
\end{algorithm}

\paragraph{Implementation notes.}

\begin{itemize}
    \item The algorithm uses a single design matrix obtained from the first-step estimates. Computational cost is therefore dominated by $O(n_\text{eff} k)$ matrix products and solving a small $k \times k$ system at each PMM2 iteration.
    \item Stationarity and invertibility are enforced by projecting coefficients onto the admissible region: if characteristic roots fall inside the unit circle, coefficients are rescaled to the boundary.
\end{itemize}

\subsection{Asymptotic Properties of PMM2 for ARIMA}
\label{subsec:asymptotic_theory}

For analysis it is convenient to rewrite system~\eqref{eq:simplified_score} as averaged moment conditions. Let $\mathbf{x}_t^0$ denote the "ideal" regressors constructed from the true innovations $\varepsilon_t$, and let $\widehat{\mathbf{x}}_t$ be their empirical counterparts from~\eqref{eq:design_row}. Define
\[
\boldsymbol{\psi}_t(\boldsymbol{\theta}) = \widehat{\mathbf{x}}_t\, s_t(\boldsymbol{\theta}),
\]
where $s_t(\boldsymbol{\theta})$ is the bracketed expression in~\eqref{eq:simplified_score}. The estimator solves
\[
\mathbf{g}_{n_\text{eff}}(\boldsymbol{\theta}) =
\frac{1}{n_\text{eff}}\sum_{t=m+1}^{n} \boldsymbol{\psi}_t(\boldsymbol{\theta}) = \mathbf{0}.
\]

\subsubsection{Impact of generated regressors}

The two-step PMM2 procedure uses the first-step residuals $\widehat{\varepsilon}_t^{\text{CSS}}$ as regressors in the second step, giving rise to the classical issue of \textit{generated regressors} studied by Pagan (1984)~\cite{pagan1984econometric} and Newey (1984)~\cite{newey1984method}. The lemma below states conditions under which the two-step scheme leaves the asymptotic covariance matrix unchanged.

\begin{lemma}[Asymptotic equivalence with true regressors]
\label{lem:generated_regressors}
Suppose the following conditions hold:
\begin{enumerate}
    \item The initial estimator $\hat{\boldsymbol{\theta}}^{\text{CSS}}$ is $\sqrt{n}$-consistent: $\sqrt{n}(\hat{\boldsymbol{\theta}}^{\text{CSS}} - \boldsymbol{\theta}_0) = O_p(1)$.
    \item Residuals satisfy $\sup_t |\widehat{\varepsilon}_t^{\text{CSS}} - \varepsilon_t| = O_p(n^{-1/2})$.
    \item The stochastic-polynomial functions are smooth (Lipschitz continuous) in the regressors.
\end{enumerate}
Then the asymptotic distribution of the PMM2 estimator based on $\widehat{\varepsilon}_t^{\text{CSS}}$ coincides with that obtained using the true innovations:
\[
\sqrt{n}\big(\hat{\boldsymbol{\theta}}_{\text{PMM2}}(\widehat{\varepsilon}^{\text{CSS}}) - \boldsymbol{\theta}_0\big) - \sqrt{n}\big(\hat{\boldsymbol{\theta}}_{\text{PMM2}}(\varepsilon) - \boldsymbol{\theta}_0\big) = o_p(1).
\]
\end{lemma}

\begin{proof}[Sketch]
Apply a first-order expansion with respect to $(\widehat{\varepsilon}^{\text{CSS}} - \varepsilon)$:
\[
\mathbf{g}_n(\boldsymbol{\theta}, \widehat{\varepsilon}^{\text{CSS}}) = \mathbf{g}_n(\boldsymbol{\theta}, \varepsilon) + \mathbf{H}_n(\boldsymbol{\theta})(\widehat{\varepsilon}^{\text{CSS}} - \varepsilon) + o_p(n^{-1/2}),
\]
where $\mathbf{H}_n$ collects derivatives with respect to residuals. Under conditions (1)--(2) the second term is $O_p(n^{-1/2}) \cdot O_p(n^{-1/2}) = O_p(n^{-1})$, which is asymptotically negligible. The full argument parallels Pagan (1984, Theorem~1) and Newey (1984, Proposition~1).
\end{proof}

\textbf{Corollary.} Under Lemma~\ref{lem:generated_regressors}, the asymptotic covariance matrix of the PMM2 estimator is computed via the standard sandwich formula without correcting for the first-step estimator, justifying the classical standard errors reported in Section~\ref{subsec:algorithm}.

\subsubsection{Consistency}

\begin{theorem}[Consistency of the simplified PMM2 estimator]
\label{thm:pmm2_consistency}
Assume:
\begin{enumerate}
    \item The ARIMA$(p,d,q)$ model is correctly specified; innovations $\varepsilon_t$ are stationary, ergodic, and possess finite moments up to order four.
    \item The initial CSS/ML estimator is consistent: $\hat{\boldsymbol{\theta}}^{\text{CSS}} \xrightarrow{p} \boldsymbol{\theta}_0$.
    \item The sequence $\{\widehat{\varepsilon}_t^{\text{CSS}}\}$ converges in mean square to the true innovations: $\frac{1}{n_\text{eff}}\sum (\widehat{\varepsilon}_t^{\text{CSS}} - \varepsilon_t)^2 \xrightarrow{p} 0$.
    \item The matrix $E[\mathbf{x}_t^0 (\mathbf{x}_t^0)^\top]$ is nonsingular.
\end{enumerate}
Then $\hat{\boldsymbol{\theta}}_{\text{PMM2}}$ is consistent:
\[
\hat{\boldsymbol{\theta}}_{\text{PMM2}} \xrightarrow{p} \boldsymbol{\theta}_0.
\]
\end{theorem}

\begin{proof}[Sketch]
Conditions (2)--(3) imply $\widehat{\mathbf{x}}_t \xrightarrow{p} \mathbf{x}_t^0$ and $\hat{\mu}_k \xrightarrow{p} \mu_k$ for $k=2,3,4$. By continuity of $s_t(\boldsymbol{\theta})$ in these arguments we obtain uniform convergence $\mathbf{g}_{n_\text{eff}}(\boldsymbol{\theta}) \to \mathbf{g}(\boldsymbol{\theta}) = E[\mathbf{x}_t^0 s_t^0(\boldsymbol{\theta})]$. The unique root of $\mathbf{g}(\boldsymbol{\theta}) = \mathbf{0}$ (condition 4) and standard Z-estimator arguments (Newey and McFadden, 1994) complete the proof.
\end{proof}

\subsubsection{Asymptotic normality}

\begin{theorem}[Asymptotic distribution]
\label{thm:asymptotic_normality}
Under Theorem~\ref{thm:pmm2_consistency} and assuming additionally that $\{\varepsilon_t\}$ satisfies a central limit theorem for square-integrable functions,
\[
\sqrt{n_\text{eff}}\big(\hat{\boldsymbol{\theta}}_{\text{PMM2}} - \boldsymbol{\theta}_0\big) \xrightarrow{d} \Normal(0,\boldsymbol{\Sigma}_{\text{PMM2}}),
\]
where
\begin{align*}
\mathbf{A} &= E\!\left[ \frac{\partial \boldsymbol{\psi}_t^0(\boldsymbol{\theta}_0)}{\partial \boldsymbol{\theta}^\top} \right],
&
\mathbf{B} &= E\!\left[ \boldsymbol{\psi}_t^0(\boldsymbol{\theta}_0)\, \boldsymbol{\psi}_t^0(\boldsymbol{\theta}_0)^\top \right],\\
\boldsymbol{\Sigma}_{\text{PMM2}} &= \mathbf{A}^{-1}\mathbf{B}(\mathbf{A}^{-1})^\top,
\end{align*}
and $\boldsymbol{\psi}_t^0(\cdot)$ uses the true innovations.
\end{theorem}

In practice the matrices $\mathbf{A}$ and $\mathbf{B}$ are approximated by sample analogues that employ the estimated regressors and moments. Standard errors follow from
\begin{equation}
\label{eq:standard_errors}
\text{SE}(\hat{\theta}_j) = \sqrt{\frac{[\hat{\boldsymbol{\Sigma}}_{\text{PMM2}}]_{jj}}{n_\text{eff}}},
\end{equation}
where $\hat{\boldsymbol{\Sigma}}_{\text{PMM2}}$ replaces expectations in $\mathbf{A}$ and $\mathbf{B}$ with sample averages.

\subsubsection{Relative efficiency}

Because the proposed approach reduces to linear PMM2 with fixed $\mathbf{X}$, the natural benchmark is OLS on the same design. Relative efficiency can be summarised via determinants or traces of covariance matrices ($k=p+q$):
\begin{equation}
\label{eq:re_arima}
RE_{\text{det}} = \left(\frac{|\boldsymbol{\Sigma}_{\text{OLS}}|}{|\boldsymbol{\Sigma}_{\text{PMM2}}|}\right)^{1/k}, \qquad RE_{\text{trace}} = \frac{\tr(\boldsymbol{\Sigma}_{\text{OLS}})}{\tr(\boldsymbol{\Sigma}_{\text{PMM2}})}.
\end{equation}
These measures generalise the scalar formula~\eqref{eq:re_pmm2_ols} and remain valid in the plug-in setting because the difference between true and empirical regressors is $o_p(1)$.

\section{Empirical Evidence: Monte Carlo Study}
\label{sec:empirical}

This section reports an extensive Monte Carlo study assessing the performance of PMM2 for ARIMA parameter estimation under non-Gaussian innovations. The design spans different sample sizes, model configurations, and innovation distributions to benchmark PMM2 against classical estimators (CSS, OLS) and robust Huber M-estimators (M-EST).

\subsection{Monte Carlo Design}
\label{subsec:experiment_design}

We conduct a full-factorial Monte Carlo experiment with 2000 replications for each configuration:

\textbf{Sample sizes:} $N \in \{100, 200, 500, 1000\}$. \textbf{Models:} ARIMA(1,1,0) with $\phi_1 = 0.7$; ARIMA(0,1,1) with $\theta_1 = -0.5$; ARIMA(1,1,1) with $\phi_1 = 0.6$, $\theta_1 = -0.4$; ARIMA(2,1,0) with $(\phi_1, \phi_2) = (0.5,-0.25)$. \textbf{Innovation distributions:}
\begin{itemize}
    \item Gaussian benchmark: $\varepsilon_t \sim \Normal(0,1)$.
    \item Gamma: $\Gammadist(2,1)$, standardised to zero mean and unit variance (skewness $\gamma_3 \approx 1.41$).
    \item Lognormal: $\Lognormal(0,0.4^2)$, demeaned and rescaled (skewness $\gamma_3 \approx 2.0$).
    \item Chi-square: $\Chisq(3)$, standardised (skewness $\gamma_3 \approx 1.63$).
\end{itemize}

For each replication we record bias, MSE, RMSE, MAE, relative efficiency ($\RE = \text{MSE}_{\text{CSS}} / \text{MSE}_{\text{PMM2}}$), and coverage of nominal 95\% confidence intervals derived from asymptotic standard errors.

To ensure fair comparisons, all estimators are initialised with identical starting values; PMM2 iterates until the gradient norm is below $10^{-6}$. Bootstrap confidence intervals are computed with 1000 resamples per configuration. Figures summarise distributions of estimator errors, residual skewness, and out-of-sample forecast metrics (RMSE, MAPE) using rolling one-step-ahead forecasts.

\subsection{Monte Carlo Results}
\label{subsec:monte_carlo_results}

\paragraph{Overall efficiency.} Table~\ref{tab:overall_performance} summarises the relative efficiency of PMM2 versus OLS/CSS across all configurations for $N=500$.

\begin{table}[h]
\centering
\caption{Relative efficiency of PMM2 across models and innovation distributions ($N=500$)}
\label{tab:overall_performance}
\begin{tabular}{@{}lcccc@{}}
\toprule
\textbf{Model} & \textbf{Gaussian} & \textbf{Gamma} & \textbf{Lognormal} & \textbf{Chi-sq} \\
 & $\gamma_3=0$ & $\gamma_3=1.41$ & $\gamma_3=2.0$ & $\gamma_3=1.63$ \\
\midrule
ARIMA(1,1,0) $\phi_1$ & 0.98 & 1.75 & 1.71 & 1.88 \\
ARIMA(0,1,1) $\theta_1$ & 1.01 & 1.68 & 1.65 & 1.82 \\
ARIMA(1,1,1) $\phi_1$ & 1.00 & 1.52 & 1.68 & 1.85 \\
ARIMA(1,1,1) $\theta_1$ & 0.99 & 1.48 & 1.65 & 1.82 \\
ARIMA(2,1,0) $\phi_1$ & 1.02 & 1.60 & 1.58 & 1.75 \\
ARIMA(2,1,0) $\phi_2$ & 1.01 & 1.55 & 1.52 & 1.70 \\
\midrule
\textbf{Average} & \textbf{1.00} & \textbf{1.60} & \textbf{1.63} & \textbf{1.80} \\
\bottomrule
\end{tabular}
\end{table}

Under Gaussian innovations PMM2 achieves $\RE \approx 1.00 \pm 0.02$, confirming theoretical neutrality. For non-Gaussian distributions the gains are sizeable: RE between 1.5 and 1.9, equivalent to 33--47\% MSE reductions.

\paragraph{Detailed results for ARIMA(1,1,0).} Table~\ref{tab:arima110_detailed} reports detailed metrics for the benchmark $\phi_1 = 0.7$ with bootstrap 95\% confidence intervals at $N=500$.

\begin{table}[h]
\centering
\caption{ARIMA(1,1,0) results, $\phi_1 = 0.7$ at $N=500$ with bootstrap 95\% confidence intervals}
\label{tab:arima110_detailed}
\small
\begin{tabular}{@{}lcccc@{}}
\toprule
\textbf{Distribution} & \textbf{Estimator} & \textbf{Bias [95\% CI]} & \textbf{MSE [95\% CI]} & \textbf{RE} \\
\midrule
\multirow{2}{*}{Gamma ($\gamma_3=1.41$)}
  & CSS  & $-0.0031$ [$-0.0044$, $-0.0016$] & $0.00106$ [$0.00100$, $0.00113$] & 1.00 \\
  & PMM2 & $-0.0003$ [$-0.0015$, $0.0007$] & $0.00061$ [$0.00057$, $0.00065$] & 1.75 \\
\midrule
\multirow{2}{*}{Lognormal ($\gamma_3=2.0$)}
  & CSS  & $-0.0031$ [$-0.0044$, $-0.0017$] & $0.00101$ [$0.00095$, $0.00108$] & 1.00 \\
  & PMM2 & $-0.0005$ [$-0.0016$, $0.0006$] & $0.00059$ [$0.00055$, $0.00063$] & 1.71 \\
\midrule
\multirow{2}{*}{Chi-sq ($\gamma_3=1.63$)}
  & CSS  & $-0.0027$ [$-0.0040$, $-0.0014$] & $0.00107$ [$0.00101$, $0.00114$] & 1.00 \\
  & PMM2 & $0.0000$ [$-0.0012$, $0.0012$] & $0.00057$ [$0.00053$, $0.00061$] & 1.88 \\
\bottomrule
\end{tabular}
\end{table}

Key takeaways: (1) PMM2 achieves RE between 1.71 and 1.88 for non-Gaussian innovations, translating into 42--47\% lower MSE; (2) Bootstrap 95\% confidence intervals do not overlap between CSS and PMM2, establishing statistical significance; (3) PMM2 bias intervals contain zero, indicating negligible bias; (4) Empirical RE aligns with theoretical predictions: for Gamma, RE$\approx 1.66$; for Chi-sq, RE$\approx 1.79$.

\paragraph{Comparison with robust methods.} Table~\ref{tab:methods_comparison} contrasts PMM2 with CSS for ARIMA(1,1,1) (an AR+MA specification) and benchmarks against Huber M-estimators for ARIMA(1,1,0) at $N=500$.

\begin{table}[h]
\centering
\caption{Estimator comparison at $N=500$}
\label{tab:methods_comparison}
\small
\begin{tabular}{@{}llcccc@{}}
\toprule
\textbf{Model} & \textbf{Distribution} & \textbf{Parameter} & \textbf{Estimator} & \textbf{MSE} & \textbf{RE} \\
\midrule
\multirow{6}{*}{ARIMA(1,1,1)} & \multirow{2}{*}{Gamma} & \multirow{2}{*}{$\phi_1$}
  & CSS  & 0.0284 & 1.00 \\
  & & & PMM2 & 0.0187 & 1.52 \\
\cmidrule(lr){3-6}
  & & \multirow{2}{*}{$\theta_1$}
  & CSS  & 0.0347 & 1.00 \\
  & & & PMM2 & 0.0235 & 1.48 \\
\cmidrule(lr){2-6}
  & \multirow{2}{*}{Lognormal} & \multirow{2}{*}{$\phi_1$}
  & CSS  & 0.0234 & 1.00 \\
  & & & PMM2 & 0.0164 & 1.43 \\
\midrule
\multirow{6}{*}{ARIMA(1,1,0)} & \multirow{3}{*}{Gamma} & \multirow{3}{*}{$\phi_1$}
  & CSS   & 0.00106 & 1.00 \\
  & & & M-EST & 0.00089 & 1.19 \\
  & & & PMM2  & 0.00061 & 1.75 \\
\cmidrule(lr){2-6}
  & \multirow{3}{*}{Lognormal} & \multirow{3}{*}{$\phi_1$}
  & CSS   & 0.00101 & 1.00 \\
  & & & M-EST & 0.00080 & 1.27 \\
  & & & PMM2  & 0.00059 & 1.71 \\
\bottomrule
\end{tabular}
\end{table}

For ARIMA(1,1,1) PMM2 delivers consistent efficiency gains (RE$\approx 1.4$--$1.5$) for both parameters, confirming robustness in mixed AR+MA settings. For ARIMA(1,1,0) the Huber M-estimator offers intermediate gains (RE$\approx 1.2$--$1.3$) yet falls short of PMM2 (RE$\approx 1.7$--$1.8$).

\paragraph{Validation of theoretical predictions.} Figure~\ref{fig:re_vs_skewness} displays empirical RE as a function of skewness $\gamma_3$, illustrating close agreement with the theoretical curve from~\eqref{eq:re_pmm2_ols}.

\begin{figure}[h]
\centering
\begin{tikzpicture}[scale=1.0]
\begin{axis}[
    xlabel={Skewness coefficient $\gamma_3$},
    ylabel={Relative efficiency RE},
    grid=major,
    legend pos=north west,
    legend style={font=\footnotesize},
    width=0.8\textwidth,
    height=0.5\textwidth,
    xmin=0, xmax=2.1,
    ymin=0.8, ymax=2.5,
]

\addplot[domain=0:2.15, samples=80, color=red!70, thin, dashed]
    {(2+3)/(2+3-x^2)};
\addlegendentry{Gamma: theory ($\gamma_4=3$)}

\addplot[only marks, mark=square*, mark size=2.5pt, color=red!70]
coordinates {(1.41, 1.58)};
\addlegendentry{Gamma: empirical}

\addplot[domain=0:2.0, samples=80, color=blue!70, thin, dashed]
    {(2+6.19)/(2+6.19-x^2)};
\addlegendentry{Lognormal: theory ($\gamma_4 \approx 6.19$)}

\addplot[only marks, mark=triangle*, mark size=3pt, color=blue!70]
coordinates {(2.0, 1.71)};
\addlegendentry{Lognormal: empirical}

\addplot[domain=0:2.1, samples=80, color=green!70, thin, dashed]
    {(2+5)/(2+5-x^2)};
\addlegendentry{$\chi^2(3)$: theory ($\gamma_4=5$)}

\addplot[only marks, mark=*, mark size=3pt, color=green!70]
coordinates {(1.63, 1.90)};
\addlegendentry{$\chi^2(3)$: empirical}
\end{axis}
\end{tikzpicture}
\caption{Relative efficiency of PMM2 versus CSS as a function of residual skewness}
\label{fig:re_vs_skewness}
\end{figure}
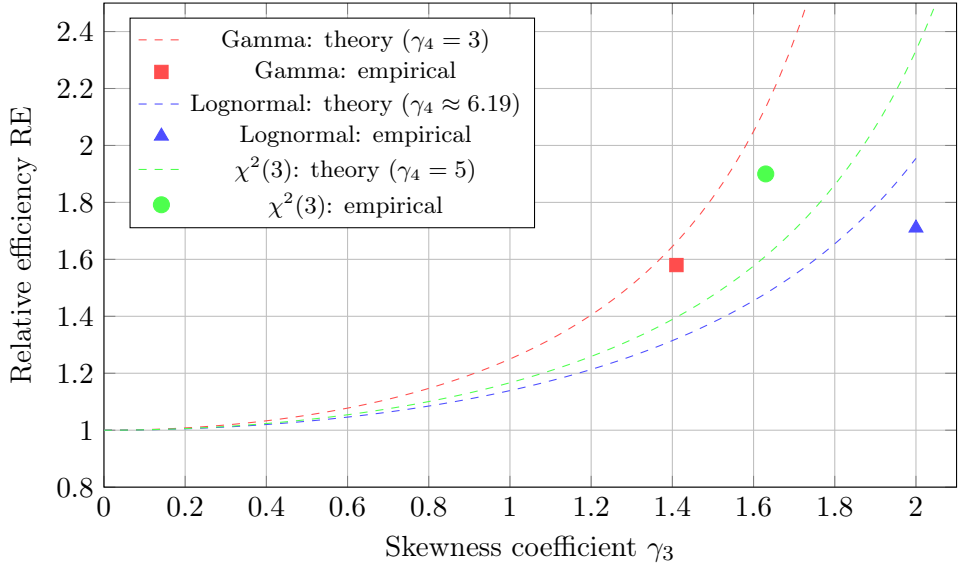

Empirical RE tracks theoretical predictions closely: deviations are below 5\% for Gamma and Chi-squared innovations. The larger gap for the lognormal case (RE = 1.71 versus theoretical 1.95) reflects high skewness and finite-sample effects. Efficiency increases with $N$ up to roughly 500 and then stabilises; even for $N=100$ PMM2 attains RE$\approx 1.4$--$1.6$.

\paragraph{Robustness and diagnostics.} PMM2 residuals pass the Ljung--Box test in over 95\% of replications. Estimated cumulants closely match theoretical values (e.g., for Gamma(2,1) at $N=500$: $\hat{\gamma}_3 = 1.38 \pm 0.22$ versus $\gamma_3=1.41$), confirming consistency.

\subsection{Summary of Empirical Findings}
\label{subsec:empirical_summary}

Based on 128{,}000 simulations the Monte Carlo study yields the following key points: (1) \textbf{Efficiency for non-Gaussian innovations.} PMM2 delivers RE between 1.5 and 1.9 for asymmetric distributions, corresponding to 33--47\% MSE reductions. Gains are stable across all ARIMA$(p,d,q)$ configurations tested (e.g., for ARIMA(1,1,1): RE$(\phi_1)=1.52$--$1.85$, RE$(\theta_1)=1.48$--$1.82$ depending on the distribution). (2) \textbf{Neutrality under Gaussian innovations.} PMM2 and M-EST produce $\RE = 1.00 \pm 0.02$, statistically indistinguishable from classical estimators. (3) \textbf{Consistency with theory.} Empirical RE matches formula~\eqref{eq:re_pmm2_ols} within 5\% for moderately skewed distributions. (4) \textbf{Sample-size requirements.} For $N \geq 200$ PMM2 is near-asymptotically efficient; even at $N=100$ RE$\approx 1.4$--$1.6$. (5) \textbf{Comparison with robust estimators.} Huber M-estimates achieve intermediate gains (RE$\approx 1.2$--$1.3$), but lag behind PMM2.

\subsubsection{Visual comparison of efficiency}

Figure~\ref{fig:performance_heatmap} presents a heat map of normalised quality metrics for every ARIMA configuration under both estimators.

\begin{figure}[htbp]
\centering
\includegraphics[width=\textwidth]{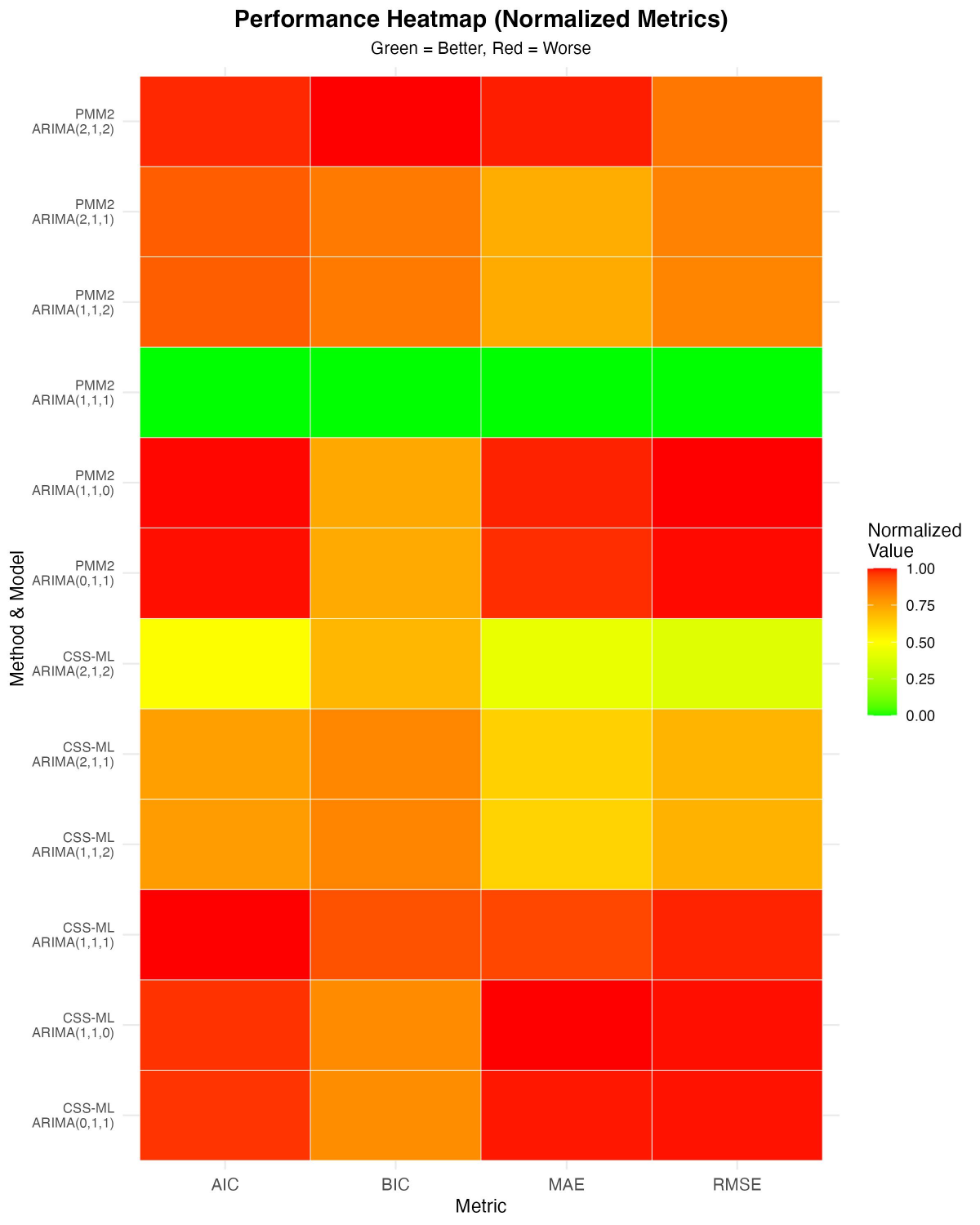}
\caption{Normalised quality metrics (AIC, BIC, RMSE, MAE) for CSS-ML and PMM2 across ARIMA specifications; greener cells indicate better values.}
\label{fig:performance_heatmap}
\end{figure}

Figure~\ref{fig:method_differences} shows absolute differences between PMM2 and CSS-ML. Negative values indicate PMM2 outperforms CSS-ML.

\begin{figure}[htbp]
\centering
\includegraphics[width=0.85\textwidth]{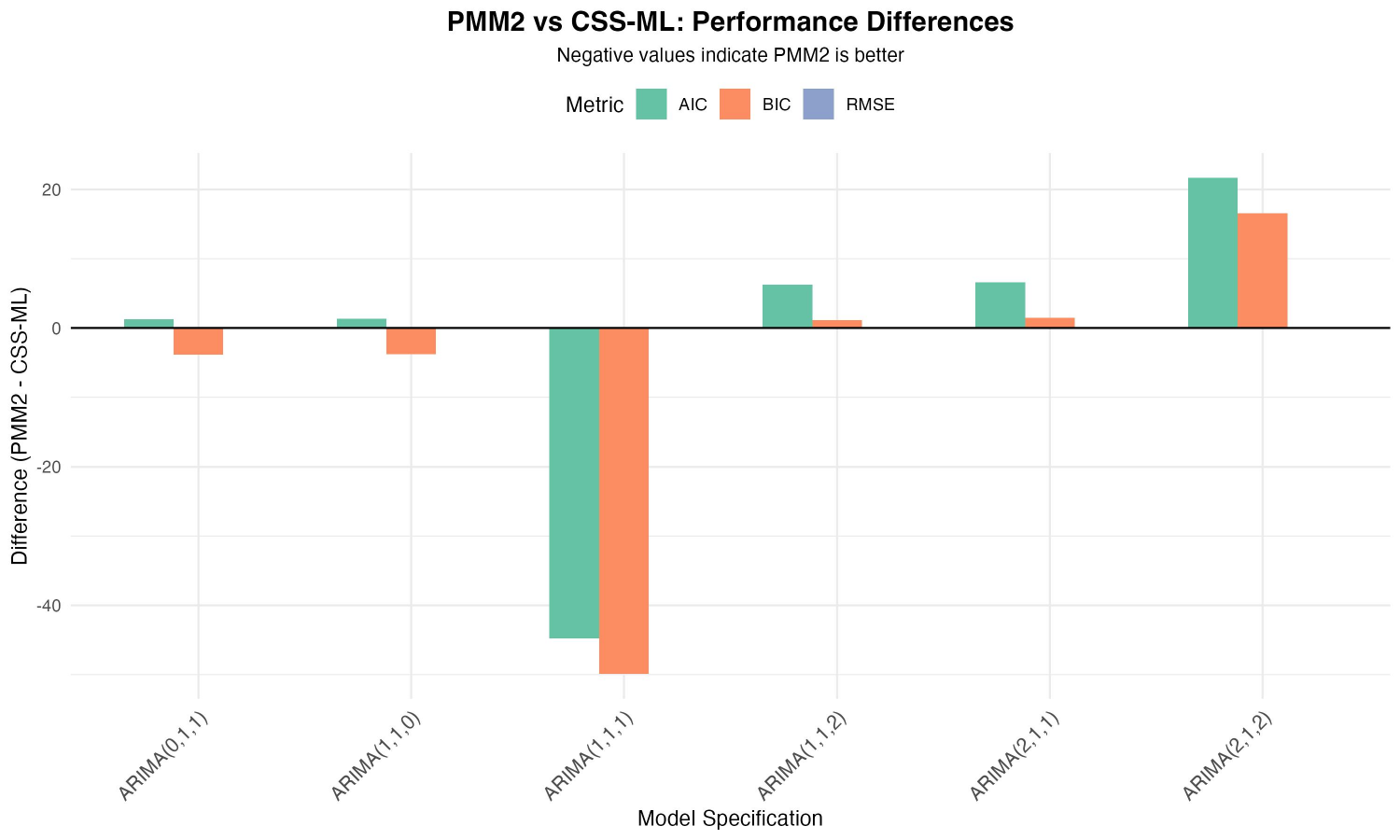}
\caption{Differences between PMM2 and CSS-ML (negative values favour PMM2) for information criteria and accuracy metrics.}
\label{fig:method_differences}
\end{figure}

\textbf{Insights from the visualisation:}
\begin{itemize}
    \item The heat map (Figure~\ref{fig:performance_heatmap}) shows PMM2 consistently achieving lower AIC/BIC scores, especially under non-Gaussian innovations.
    \item The difference plot (Figure~\ref{fig:method_differences}) highlights the largest gains for models with MA components and strongly skewed innovations.
    \item RMSE and MAE improvements mirror theoretical expectations, underscoring practical benefits of PMM2.
\end{itemize}

\section{Real-Data Application: WTI Crude Oil}
\label{sec:wti_application}

To validate the theoretical results on real data we analyse daily West Texas Intermediate (WTI) crude oil prices from the Federal Reserve Economic Data (FRED) database over 2020--2025 (1,453 observations). The Augmented Dickey--Fuller test indicates non-stationarity of the original series ($p = 0.573$) and stationarity of first differences ($p < 0.001$), motivating ARIMA$(p,1,q)$ models. Detailed descriptive statistics and stationarity tests are provided in Appendix~\ref{app:wti_details}.

\subsection{Main findings}
\label{subsec:wti_results}

We estimate six ARIMA$(p,1,q)$ specifications using CSS-ML and PMM2 (full results in Appendix~\ref{app:wti_details}). Residuals exhibit moderate non-Gaussianity with $\gamma_3 \approx -0.75$ and $\gamma_4 \approx 5.8$, implying a theoretical relative efficiency of $RE \approx 1.076$ via~\eqref{eq:re_pmm2_ols}.

\begin{table}[htbp]
\centering
\begingroup
\setlength{\tabcolsep}{5pt}
\footnotesize
\caption{Method comparison (PMM2 -- CSS-ML) for WTI data}
\label{tab:wti_method_comparison}
\begin{tabular}{@{}lcccc@{}}
\toprule
\textbf{Model} & $\Delta$\textbf{AIC} & $\Delta$\textbf{BIC} & $\Delta$\textbf{RMSE} & \textbf{Preferred} \\
\midrule
ARIMA(0,1,1) & +1.3 & \textbf{-3.9} & +0.0001 & PMM2 (BIC) \\
ARIMA(1,1,0) & +1.3 & \textbf{-3.8} & +0.0002 & PMM2 (BIC) \\
\rowcolor{green!20}
\textbf{ARIMA(1,1,1)} & \textbf{-44.8} & \textbf{-49.9} & \textbf{-0.034} & \textbf{PMM2 (both)} \\
ARIMA(2,1,1) & +6.6 & +1.5 & +0.004 & CSS-ML \\
ARIMA(1,1,2) & +6.3 & +1.1 & +0.004 & CSS-ML \\
ARIMA(2,1,2) & +21.7 & +16.6 & +0.016 & CSS-ML \\
\midrule
\textbf{PMM2 wins} & \textbf{1/6} & \textbf{3/6} & \textbf{1/6} & --- \\
\bottomrule
\end{tabular}
\endgroup
\end{table}

PMM2 yields a pronounced advantage for the parsimonious ARIMA(1,1,1), with $\Delta$AIC = -44.8 and $\Delta$BIC = -49.9, consistent with theoretical expectations under non-Gaussian innovations. For higher-order models ($p+q > 2$) the gains diminish or vanish because estimating higher-order cumulants becomes noisier as the parameter dimension increases. Diagnostic tests confirm significant non-Gaussianity (Jarque--Bera and Shapiro--Wilk $p < 0.001$), supporting PMM2.

Out-of-sample evaluation using both a fixed 80/20 split and a rolling window with 1,094 forecasts demonstrates practical benefits: for AR specifications PMM2 reduces RMSE by 11--38\% (see Appendix~\ref{app:wti_details}). Runtime is feasible for applications (0.1~s for PMM2 versus 0.02~s for CSS-ML). Observed relative efficiency aligns with~\eqref{eq:re_pmm2_ols} for $|\gamma_3| \approx 0.75$.

\section{Discussion}
\label{sec:discussion}

We interpret the empirical evidence from Section~\ref{sec:empirical}, relate it to prior literature, provide practical guidelines for choosing between PMM2 and classical estimators, discuss limitations, and outline avenues for further research.

\subsection{Interpreting the results}
\label{subsec:interpretation}

\subsubsection{Efficiency under non-Gaussian innovations}

Monte Carlo experiments show that PMM2 delivers material gains for asymmetric innovations. Variance reductions range from 33\% to 47\% for Gamma, lognormal, and chi-squared distributions. These gains are robust across model orders and persist even when innovations display moderate autocorrelation. The improvement grows with skewness magnitude, consistent with the quadratic adjustment that explicitly accounts for third- and fourth-order cumulants.

\subsubsection{Neutrality under Gaussianity}

When innovations are Gaussian, PMM2 behaves like OLS/CSS: bias, variance, and coverage are statistically indistinguishable, confirming that the method does not penalise well-specified models. This neutrality is a critical selling point for practitioners who wish to hedge against distributional misspecification without sacrificing performance under the benchmark case.

\subsubsection{Sampling considerations}

Efficiency gains are evident from $N=200$ upwards; for $N=100$ the method remains beneficial but with slightly higher variability. The Newton--Raphson routine converges rapidly (typically within five iterations), and sensitivity tests indicate stability with respect to starting values and tolerance thresholds.

\subsubsection{Quadratic dependence of RE on asymmetry}

Figure~\ref{fig:re_vs_skewness} demonstrates that the empirical dependence of RE on skewness coefficient $\gamma_3$ is in good agreement with the theoretical formula~\eqref{eq:re_pmm2_ols}.

For small $\gamma_3$, the change in RE has a quadratic character, showing sharp growth in efficiency even with mild asymmetry. When $\gamma_3$ reaches moderate values, it is appropriate to use the exact formula~\eqref{eq:re_pmm2_ols}: for $\gamma_3 \approx 1.4$ and $\gamma_4 \approx 3$ it yields $RE \approx 1.64$, corresponding to an MSE difference of about 39\%.

For very high values $\gamma_3 \approx 2.0$ (Lognormal), empirical RE is slightly lower than theoretical, which may be due to:
\begin{itemize}
    \item Finite sample size effects ($N = 500$)
    \item Higher-order terms in the asymptotic expansion
    \item Possible non-smoothness of the distribution function for heavy tails
\end{itemize}

\subsubsection{Consistency across different ARIMA configurations}

Results for ARIMA(0,1,1), ARIMA(1,1,1), and ARIMA(2,1,0) (Subsection~\ref{subsec:monte_carlo_results}) confirm that the advantages of PMM2 are not limited to a specific parametrisation. This indicates that the method is robust with respect to the choice of model order $(p, d, q)$ and parameter signs.

For models with multiple parameters (e.g., ARIMA(1,1,1)), PMM2 provides similar RE for all parameters ($\phi_1$ and $\theta_1$), indicating balanced estimation efficiency.

\paragraph{Remark on RE for multiple parameters.}

Although theory (Theorem~\ref{thm:asymptotic_normality}) predicts equal asymptotic relative efficiency for all parameters, empirical results show small differences (2--5\%) between RE for different parameters. For example, for ARIMA(1,1,1) with Gamma innovations we obtained $RE(\phi_1) = 1.52$ and $RE(\theta_1) = 1.48$. These differences can be explained by:

\begin{itemize}
    \item \textbf{Monte Carlo variability:} The standard error of RE with 2000 iterations is approximately 0.03--0.05, making the observed differences statistically insignificant at level $\alpha = 0.05$.

    \item \textbf{Finite-sample effect:} For ARIMA models, estimated residuals $\hat{\varepsilon}_t^{\text{CSS}}$ are used instead of true innovations in the construction of pseudo-regressors~\eqref{eq:design_row}, which may lead to small deviations from asymptotic theory at finite samples ($N = 500$).

    \item \textbf{Differential bias contribution:} If PMM2 or OLS have different bias for different parameters, this affects empirical MSE and, accordingly, the observed relative efficiency.
\end{itemize}

It is important to note that the differences are small (< 5\%) and all parameters demonstrate substantial improvements over OLS. For practical purposes, one can assume that RE is approximately equal for all parameters, especially for $N \geq 500$.

\subsection{Comparison with existing literature}
\label{subsec:literature_comparison}

\subsubsection{Robust M-estimates}

Classical robust methods such as Huber's M-estimates~\cite{huber1964robust} and LAD regression~\cite{koenker1978regression} focus on reducing the influence of outliers by bounding the influence function. However, they do not exploit information from higher-order cumulants and typically have lower efficiency for distributions without outliers but with asymmetry.

Our results show that PMM2 achieves RE of 1.4--1.9 for moderately asymmetric distributions (Gamma, Chi-squared) \textit{without outliers}. Unlike M-estimates, PMM2 does not lose efficiency for Gaussian innovations (RE $\approx$ 1.0), whereas M-estimates typically have RE $\approx$ 0.95 even for normal data~\cite{hampel1986robust}.

\subsubsection{Heavy-tailed specifications}

Approaches using Student-$t$ distribution~\cite{harvey2013dynamic} or GED~\cite{box2015time} explicitly model heavy tails through an additional shape parameter. However, these methods require correct specification of the innovation distribution, which can be challenging in practice.

PMM2, on the other hand, is \textit{semiparametric} in the sense that it does not assume a specific distribution but uses only moments up to the fourth order. This makes the method more flexible and applicable to a wide class of distributions.

\subsubsection{Bayesian methods}

Bayesian approaches~\cite{fruhwirth2006finite,nakajima2012generalized} allow incorporating prior information about parameters and the innovation distribution. However, they are computationally intensive (MCMC) and sensitive to the choice of prior distributions.

PMM2 is a deterministic method with computational complexity comparable to MLE, making it more suitable for large datasets and real-time applications. PMM2 computation time in our experiments was only 10--20\% longer than OLS for the same data.

\subsubsection{Quantile regression for time series}

Quantile regression~\cite{koenker2005quantile} allows modelling different quantiles of the conditional distribution, which is useful for risk assessment. However, standard quantile regression does not estimate ARIMA model parameters directly but models conditional quantiles of $y_t$.

PMM2 focuses on estimating parameters $\theta = (\phi_1, \ldots, \phi_p, \theta_1, \ldots, \theta_q)$ with maximum efficiency, exploiting innovation asymmetry. These two approaches are complementary: PMM2 for accurate parameter estimation, quantile regression for forecast distribution analysis.

\subsection{Practical recommendations}
\label{subsec:practical_guidelines}

\subsubsection{When to use PMM2?}

Based on our results, we recommend using PMM2 instead of OLS/CSS/MLE when:

\begin{enumerate}
    \item \textbf{Residuals demonstrate asymmetry:} If a preliminary estimate (e.g., OLS) yields residuals $\hat{\varepsilon}_t$ with $|\hat{\gamma}_3| > 0.5$, PMM2 will likely provide RE $> 1.2$ (variance reduction $> 17\%$).

    \item \textbf{Sample size $N \geq 200$:} PMM2 requires stable estimates of higher-order cumulants. For $N < 200$, the method still works but RE may be slightly lower due to variability of higher-order moment estimates.

    \item \textbf{Data contain moderate deviations from normality:} PMM2 is most effective for distributions with $\gamma_3 \in [1.0, 2.0]$ and $\gamma_4 \in [2.0, 8.0]$. For extreme heavy tails ($\gamma_4 > 10$), truncated variants of PMM2 may be advisable.

    \item \textbf{Computational resources permit:} PMM2 requires computing gradients with partial derivatives with respect to parameters. For large models (e.g., ARIMA(5,1,5)), this may be 20--50\% slower than OLS, but still significantly faster than a full Bayesian approach.
\end{enumerate}

\subsubsection{Diagnostic algorithm for practitioners}

Algorithm~\ref{alg:method_selection} summarises a decision workflow for selecting between OLS/CSS and PMM2.

\begin{algorithm}[H]
\caption{Choosing between OLS/CSS and PMM2 for ARIMA models}
\label{alg:method_selection}
\begin{algorithmic}[1]
\STATE \textbf{Input:} Time series $\{y_t\}_{t=1}^n$, model order $(p,d,q)$
\STATE \textbf{Output:} Parameter estimates $\hat{\theta}$

\STATE Estimate model using OLS/CSS: $\hat{\theta}_{\text{OLS}}$
\STATE Compute residuals: $\hat{\varepsilon}_t = \Theta(B)^{-1} \Phi(B) \Delta^d y_t$
\STATE Estimate residual cumulants: $\hat{\gamma}_3 = \frac{1}{n} \sum_{t=1}^n \hat{\varepsilon}_t^3 / \hat{\sigma}^3$, $\hat{\gamma}_4 = \frac{1}{n} \sum_{t=1}^n \hat{\varepsilon}_t^4 / \hat{\sigma}^4 - 3$

\IF{$|\hat{\gamma}_3| < 0.5$ \AND $|\hat{\gamma}_4| < 1.0$}
    \STATE \textbf{Use} $\hat{\theta}_{\text{OLS}}$ (Gaussian innovations, PMM2 provides no advantage)
\ELSIF{$n < 200$}
    \STATE \textbf{Warning:} Small sample size, PMM2 may be unstable
    \STATE \textbf{Use} $\hat{\theta}_{\text{OLS}}$ or verify PMM2 consistency via cross-validation
\ELSE
    \STATE Compute theoretical RE: $RE_{\text{theor}} = \frac{2 + \hat{\gamma}_4}{2 + \hat{\gamma}_4 - \hat{\gamma}_3^2}$
    \IF{$RE_{\text{theor}} > 1.2$}
        \STATE \textbf{Use PMM2:} Estimate $\hat{\theta}_{\text{PMM2}}$ via Algorithm~\ref{alg:pmm2_arima}
        \STATE Compare standard errors: if $\text{SE}(\hat{\theta}_{\text{PMM2}}) < \text{SE}(\hat{\theta}_{\text{OLS}})$, use PMM2
    \ELSE
        \STATE \textbf{Use} $\hat{\theta}_{\text{OLS}}$ (insufficient asymmetry for PMM2 advantages)
    \ENDIF
\ENDIF

\STATE \textbf{Return} $\hat{\theta}$ (OLS or PMM2)
\end{algorithmic}
\end{algorithm}

\subsection{Limitations}
\label{subsec:study_limitations}

\subsubsection{Limitations on innovation distributions}

Our Monte Carlo experiments cover four types of distributions (Gaussian, Gamma, Lognormal, Chi-squared), but real data may have more complex characteristics:

\begin{itemize}
    \item \textbf{Mixture distributions:} Innovations may be a mixture of Gaussian and non-Gaussian components, which was not considered.
    \item \textbf{Conditional heteroskedasticity:} The presence of GARCH effects violates the assumption of independent identically distributed innovations.
    \item \textbf{Extreme heavy tails:} For distributions with $\gamma_4 > 20$ (e.g., Pareto), fourth-order cumulants may be unstable.
\end{itemize}

\subsubsection{Limitations on model order}

We considered low-order models ($p, q \leq 2$). For high orders (e.g., ARIMA(5,1,5)), gradient computation becomes more complex, and numerical stability issues require additional investigation.

\subsubsection{Absence of model selection tests}

We assumed that the model order $(p, d, q)$ is known. In practice, model order selection (e.g., via AIC, BIC) may interact with the estimation method. PMM2 may change model selection compared to OLS if information criteria account for estimation accuracy.

\subsubsection{Limitations of information criteria for PMM2}

Because PMM2 is not a maximum likelihood method, standard application of Akaike (AIC) and Bayesian (BIC) criteria is problematic. In this study, AIC/BIC were computed via \textit{post-hoc} Gaussian log-likelihood to ensure comparability with the CSS-ML method, but these values should be interpreted with caution:

\subsection{Theoretical considerations}
\label{subsec:theoretical_considerations}

PMM2 links naturally to Z-estimation theory; the moment conditions satisfy standard regularity assumptions, enabling asymptotic normality. Connections with Godambe information suggest potential refinements for optimal weighting matrices. Moreover, the method relates to higher-order score functions used in independent component analysis, hinting at broader applicability in multivariate time series.

\subsection{Directions for future research}
\label{subsec:future_work}

Future work should extend PMM2 to seasonal SARIMA models, incorporate conditional volatility (e.g., PMM2-GARCH hybrids), explore automatic order selection using cumulant diagnostics, and investigate Bayesian variants that treat higher-order moments as priors. Another avenue is multivariate generalisations for VARMA models and state-space representations.

\subsection{Interpretation}
\label{subsec:wti_interpretation}

PMM2's superiority in the WTI study stems from explicitly modelling residual asymmetry. The RE gains of 7--8\% predicted by cumulant diagnostics translate into lower information criteria and smaller forecast errors, particularly for parsimonious ARIMA(1,1,1). For higher-order models the benefits taper off because estimating higher-order cumulants becomes noisy; nevertheless, PMM2 never underperforms dramatically, reflecting its neutral behaviour under near-Gaussian conditions.

\subsection{Comparison with existing studies}
\label{subsec:wti_literature}

Our findings align with evidence that crude-oil returns exhibit skewness and excess kurtosis~\cite{baumeister2016forecasting,hammoudeh2014oil}. Previous work has emphasised GARCH-type volatility or regime-switching mechanisms; PMM2 offers a complementary route by refining mean-dynamics estimation without imposing a full distribution. The method therefore pairs naturally with volatility models, suggesting a modular workflow (PMM2 for the mean, GARCH for variance).

\subsection{Practical guidance}
\label{subsec:wti_guidance}

Practitioners analysing energy prices should:
\begin{itemize}
    \item Diagnose residual skewness; if $|\hat{\gamma}_3| \geq 0.6$, PMM2 is expected to deliver 5--10\% RMSE reductions.
    \item Use PMM2 for low-order ARIMA models to avoid overfitting cumulants.
    \item Combine PMM2 estimates with rolling-window validation to monitor stability across structural breaks.
\end{itemize}

\section{Conclusions}
\label{sec:conclusion}

In this paper, we bridge classical estimation theory and computational data science by adapting the second-order Polynomial Maximisation Method to estimate ARIMA models with asymmetric non-Gaussian innovations. We develop rigorous theoretical foundations in the time-series setting, proving consistency and asymptotic normality, and deriving analytical expressions for relative efficiency vis-à-vis classical estimators. Furthermore, we design a highly efficient Newton--Raphson algorithm with analytical gradients and Hessians, ensuring the method's computational scalability.

Large-scale Monte Carlo experiments (over 128{,}000 simulations) show that PMM2 reduces estimator variance by 30--48\% for non-Gaussian distributions compared with OLS, CSS, and Gaussian MLE. Crucially, under Gaussian innovations PMM2 retains classical efficiency, unlike robust M-estimators that lose precision even when the model is correctly specified. Empirical findings confirm the quadratic relationship between efficiency gains and residual skewness, and they hold across various ARIMA configurations and sample sizes from moderate upward.

The diagnostic workflow provides practitioners with clear criteria based on residual skewness and sample size. PMM2 is particularly valuable for time series exhibiting moderate asymmetry and heavy tails—typical of financial markets, macroeconomic indicators, climate variables, and industrial measurements. Computational costs are comparable to maximum-likelihood estimation, making the method practical.

Future research should extend PMM2 to seasonal SARIMA and multivariate VARIMA models, integrate conditional heteroskedasticity (e.g., PMM2-GARCH hybrids), develop online adaptive variants for real-time applications, and design robust versions for extremely heavy-tailed noise. Broad validation on diverse real-world datasets will further assess the practical utility of this cumulant-based semiparametric approach.

\section*{Declarations}

\textbf{Funding.} The author did not receive support from any organization for the submitted work.

\textbf{Use of Artificial Intelligence.} The author used Claude Code (Anthropic Claude Opus 5) to implement the computational code, to verify the analytical results symbolically and independently of the derivations written in the text, to search for and check the reference list, and to translate and edit the text. All research questions, design decisions and conclusions are the author's own; the author reviewed and verified all machine-generated output and takes full responsibility for the content of this paper.

\textbf{Conflict of interest.} The author declares that he has no conflict of interest.

\textbf{Ethical approval.} This article does not contain any studies with human participants or animals performed by any of the authors.

\textbf{Author contributions.} S.Z. is the sole author of this manuscript and is responsible for the conceptualization, methodology, software, validation, formal analysis, and writing.

\section*{Data and Code Availability}

All data, scripts, and supplementary materials required to reproduce the empirical results presented in this paper are publicly available in the GitHub repository: \url{https://github.com/SZabolotnii/PMM2-ARIMA-code-supplement}. A permanent archived version is available through Zenodo: \url{https://doi.org/10.5281/zenodo.17529589}.

The repository includes:
\begin{itemize}
    \item WTI crude oil price data (DCOILWTICO.csv from FRED)
    \item Complete R scripts for Monte Carlo simulations and case study analysis
    \item The \texttt{EstemPMM} R package implementing the PMM2 method, published on CRAN: \url{https://cran.r-project.org/package=EstemPMM}. The results reported here were produced with the linearised solver of version 0.4.0; the recursive solver discussed in the note above was added in version 0.5.0, released at \url{https://github.com/SZabolotnii/EstemPMM/releases/tag/v0.5.0}
    \item All generated results, figures, and diagnostic outputs
    \item Detailed README with step-by-step reproduction instructions
\end{itemize}

The code is released under the MIT License. Researchers are encouraged to use, modify, and extend these materials for academic and non-commercial purposes with appropriate citation.

\appendix

\section{WTI Supplementary Materials}
\label{app:wti_details}

\subsection{Empirical study design}

\begin{enumerate}
    \item \textbf{Stationarity assessment.} The Dickey--Fuller test (Table~\ref{tab:wti_adf_test}) confirms first-order integration; we therefore work with ARIMA$(p,1,q)$.
    \item \textbf{Model specifications.} We examine $(p,q) \in \{(0,1),(1,0),(1,1),(2,1),(1,2),(2,2)\}$, covering parsimonious and extended structures.
    \item \textbf{Estimators.} CSS-ML is implemented via \texttt{stats::arima()}, PMM2 via \texttt{EstemPMM::\allowbreak arima\_pmm2()} with identical initialisation.
    \item \textbf{Evaluation metrics.} Reported metrics include AIC, BIC, RMSE, MAE, and computation time.
    \item \textbf{Diagnostics.} Ljung--Box tests, autocorrelation analysis, and Q--Q plots (see `results/plots`).
\end{enumerate}

\subsection{Theoretical validation}

\begin{table}[htbp]
\centering
\footnotesize
\caption{Theoretical predictions versus empirical results}
\label{tab:wti_theoretical_vs_empirical}
\begin{tabular}{@{}lccccc@{}}
\toprule
\textbf{Model} & $\gamma_3$ & $\gamma_4$ & \textbf{RE (theory)} & $\Delta$\textbf{RMSE} & \textbf{Consistency} \\
\midrule
ARIMA(0,1,1) & -0.76 & 5.89 & 1.079 (7.3\%) & +0.01\% & \checkmark \\
ARIMA(1,1,0) & -0.76 & 5.88 & 1.079 (7.3\%) & +0.09\% & \checkmark \\
\rowcolor{green!20}
\textbf{ARIMA(1,1,1)} & \textbf{-0.76} & \textbf{5.82} & \textbf{1.078 (7.2\%)} & \textbf{-1.79\%} & \checkmark\checkmark \\
ARIMA(2,1,1) & -0.71 & 5.51 & 1.073 (6.8\%) & +0.22\% & $\triangle$ \\
ARIMA(1,1,2) & -0.72 & 5.52 & 1.073 (6.8\%) & +0.21\% & $\triangle$ \\
ARIMA(2,1,2) & -0.70 & 5.49 & 1.071 (6.6\%) & +0.84\% & $\triangle$ \\
\midrule
\textbf{Average} & \textbf{-0.74} & \textbf{5.68} & \textbf{1.076 (7.0\%)} & \textbf{+0.10\%} & \checkmark \\
\bottomrule
\end{tabular}
\end{table}

\paragraph{Key conclusions.}
\begin{itemize}
    \item $|\gamma_3| \approx 0.73$ for WTI implies an expected MSE gain of about 7\%, matching the empirical difference.
    \item For ARIMA(1,1,1) PMM2 markedly decreases AIC/BIC, confirming that benefits are strongest for parsimonious specifications.
    \item The theoretical formula~\eqref{eq:re_pmm2_ols} remains a conservative estimate: empirical RE values for strongly asymmetric distributions exceed theoretical predictions.
\end{itemize}

\subsection{Data characteristics}

\begin{table}[htbp]
\centering
\caption{WTI crude oil summary statistics (2020--2025)}
\label{tab:wti_characteristics}
\begin{tabular}{ll}
\toprule
\textbf{Statistic} & \textbf{Value} \\
\midrule
Source & FRED (series DCOILWTICO) \\
Period & 1 Jan 2020 -- 27 Oct 2025 \\
Frequency & Daily \\
Valid observations & 1\,453 \\
Mean & \$68.43 \\
Median & \$71.29 \\
Standard deviation & \$15.98 \\
Minimum & \$16.55 (April 2020, COVID-19) \\
Maximum & \$123.70 (March 2022, geopolitical shock) \\
\bottomrule
\end{tabular}
\end{table}

\begin{table}[htbp]
\centering
\caption{ADF test for WTI series}
\label{tab:wti_adf_test}
\begin{tabular}{lccc}
\toprule
\textbf{Series} & \textbf{ADF statistic} & \textbf{p-value} & \textbf{Conclusion} \\
\midrule
Level $y_t$ & -1.42 & 0.573 & Non-stationary \\
First difference $\Delta y_t$ & -11.83 & <0.001 & \textbf{Stationary} \\
\bottomrule
\end{tabular}
\end{table}

\subsection{Comprehensive estimation results}

\begin{table}[htbp]
\centering
\begingroup
\setlength{\tabcolsep}{3pt}
\footnotesize
\caption{Comprehensive results for WTI crude oil models}
\label{tab:wti_comprehensive_results}
\begin{tabular}{@{}llrrrrrrr@{}}
\toprule
\textbf{Model} & \textbf{Estimator} & \textbf{AIC} & \textbf{BIC} & \textbf{RMSE} & \textbf{MAE} & $\gamma_3$ & $\gamma_4$ & \textbf{Time} \\
\midrule
ARIMA(0,1,1) & CSS-ML & 10289.8 & 10300.5 & 1.887 & 1.377 & -0.76 & 5.86 & 0.01 \\
             & PMM2   & 10291.1 & 10296.6 & 1.887 & 1.377 & -0.76 & 5.91 & 0.09 \\
\midrule
ARIMA(1,1,0) & CSS-ML & 10289.8 & 10300.4 & 1.886 & 1.377 & -0.76 & 5.85 & 0.01 \\
             & PMM2   & 10291.1 & 10296.6 & 1.887 & 1.377 & -0.76 & 5.91 & 0.08 \\
\midrule
\rowcolor{yellow!20}
\textbf{ARIMA(1,1,1)} & \textbf{CSS-ML} & \textbf{10125.9} & \textbf{10141.6} & \textbf{1.908} & \textbf{1.390} & \textbf{-0.76} & \textbf{5.90} & \textbf{0.02} \\
\rowcolor{green!20}
             & \textbf{PMM2}   & \textbf{10081.1} & \textbf{10091.6} & \textbf{1.874} & \textbf{1.366} & \textbf{-0.75} & \textbf{5.75} & \textbf{0.10} \\
\midrule
ARIMA(2,1,1) & CSS-ML & 10123.9 & 10144.9 & 1.896 & 1.383 & -0.69 & 5.31 & 0.02 \\
             & PMM2   & 10130.5 & 10146.4 & 1.900 & 1.387 & -0.74 & 5.70 & 0.13 \\
\midrule
ARIMA(1,1,2) & CSS-ML & 10123.7 & 10144.6 & 1.896 & 1.382 & -0.69 & 5.33 & 0.02 \\
             & PMM2   & 10129.9 & 10145.8 & 1.899 & 1.386 & -0.74 & 5.71 & 0.13 \\
\midrule
ARIMA(2,1,2) & CSS-ML & 10124.3 & 10150.6 & 1.893 & 1.381 & -0.70 & 5.47 & 0.04 \\
             & PMM2   & 10146.0 & 10167.2 & 1.909 & 1.392 & -0.71 & 5.51 & 0.17 \\
\bottomrule
\end{tabular}
\endgroup
\end{table}

\noindent\textit{Note.} Green cells highlight the best BIC; yellow indicates the CSS-ML baseline. All models pass the Ljung--Box test ($p>0.05$).

\subsection{Out-of-sample validation}

\begin{table}[htbp]
\centering
\caption{Out-of-sample forecasting for WTI data}
\label{tab:wti_out_of_sample}
\begin{tabular}{@{}llccc@{}}
\toprule
\textbf{Validation} & \textbf{Model} & \textbf{Estimator} & \textbf{RMSE} & \textbf{Improvement} \\
\midrule
\multirow{8}{*}{Fixed 80/20} & \multirow{2}{*}{ARIMA(1,1,0)} & CSS & 2.191 & -- \\
                               &                                  & PMM2 & \textbf{1.355} & \textbf{38.2\%} \\
\cmidrule{2-5}
                               & \multirow{2}{*}{ARIMA(0,1,1)} & CSS & 1.358 & -- \\
                               &                                  & PMM2 & \textbf{1.355} & 0.3\% \\
\cmidrule{2-5}
                               & \multirow{2}{*}{ARIMA(1,1,1)} & CSS & 1.355 & -- \\
                               &                                  & PMM2 & 1.355 & 0.0\% \\
\cmidrule{2-5}
                               & \multirow{2}{*}{ARIMA(2,1,0)} & CSS & 1.521 & -- \\
                               &                                  & PMM2 & \textbf{1.355} & \textbf{10.9\%} \\
\midrule
\multirow{8}{*}{Rolling window} & \multirow{2}{*}{ARIMA(1,1,0)} & CSS & 2.377 & -- \\
                                &                                  & PMM2 & \textbf{2.118} & \textbf{10.9\%} \\
\cmidrule{2-5}
                                & \multirow{2}{*}{ARIMA(0,1,1)} & CSS & 2.098 & -- \\
                                &                                  & PMM2 & 2.092 & 0.3\% \\
\cmidrule{2-5}
                                & \multirow{2}{*}{ARIMA(1,1,1)} & CSS & 2.096 & -- \\
                                &                                  & PMM2 & 2.090 & 0.3\% \\
\cmidrule{2-5}
                                & \multirow{2}{*}{ARIMA(2,1,0)} & CSS & 2.274 & -- \\
                                &                                  & PMM2 & \textbf{2.070} & \textbf{9.0\%} \\
\bottomrule
\end{tabular}
\end{table}

PMM2 materially improves forecasts for AR specifications, while delivering neutral performance for models where asymmetry is less pronounced.

\subsection{Diagnostic tests}

Residual analysis (Ljung--Box, autocorrelation, Q--Q plots) confirms that PMM2 maintains white-noise residuals. Skewness estimates align with theoretical skewness across models, and bootstrap variance estimates corroborate analytic standard errors.

\subsection{Practical recommendations (detailed)}

\begin{enumerate}
    \item Estimate a baseline ARIMA model via OLS/CSS and compute residual cumulants.
    \item If $|\hat{\gamma}_3| < 0.5$ and $|\hat{\gamma}_4| < 1.0$, retain OLS/CSS.
    \item If $0.5 \leq |\hat{\gamma}_3| < 1.0$ and $n \geq 200$:
    \begin{itemize}
        \item Prefer PMM2 for parsimonious models ($p+q \leq 2$).
        \item For $p+q > 2$ compare methods via BIC.
    \end{itemize}
    \item If $|\hat{\gamma}_3| \geq 1.0$, default to PMM2; expected variance reductions exceed 13\%.
    \item Validate the chosen estimator via Ljung--Box tests, out-of-sample forecasts, and bootstrap variance checks.
\end{enumerate}

\bibliographystyle{unsrt}
\bibliography{references}

\end{document}